\documentclass[11pt,letterpaper,english]{article}

\usepackage[T1]{fontenc}
\usepackage[utf8]{inputenc}
\usepackage[english]{babel}
\usepackage[margin=1in]{geometry}
\usepackage[numbers]{natbib}
\setcitestyle{nocompress}
\bibpunct{(}{)}{;}{a}{}{,}

\usepackage{graphicx}
\usepackage[cmex10]{amsmath}
\usepackage{authblk}
\usepackage{macros}

\begin{document}
\title{Simultaneous Perturbation Methods for Adaptive Labor Staffing in
Service Systems}
\author[$\dag$]{Prashanth L.A.}
\author[$\sharp$]{H.L. Prasad}
\author[$\$$]{Nirmit Desai}
\author[$\sharp$]{Shalabh Bhatnagar}
\author[$\$$]{Gargi Dasgupta}
\affil[$^\dag$]{SequeL Team, INRIA Lille - Nord Europe, FRANCE}
\affil[$^\sharp$]{Department of Computer Science and Automation, Indian Institute of Science, INDIA}
\affil[$^\sharp$]{IBM Research, INDIA}

\renewcommand\Authands{ and }

\date{}

\maketitle

\begin{abstract}
Service systems are labor intensive due to the large variation in the
tasks required to address service requests from multiple customers.
Aligning the staffing levels to the forecasted workloads adaptively in
such systems is nontrivial because of a large number of parameters and
operational variations leading to a huge search space.  A
challenging problem here is to optimize the staffing while maintaining the
system in steady-state and compliant to aggregate service level agreement (SLA) constraints.
Further, because these parameters change on a weekly basis, the
optimization should not take longer than a few hours.  We formulate
this problem as a constrained Markov cost process parameterized by the (discrete) staffing levels. We propose novel simultaneous perturbation stochastic approximation (SPSA) based SASOC (Staff Allocation using Stochastic Optimization with Constraints) algorithms for solving the above problem. The algorithms include both first order as well as second order methods and incorporate SPSA based gradient estimates in the primal, with dual ascent for the Lagrange multipliers. Both the algorithms that we propose are online, incremental and easy to implement. Further, they involve a certain generalized smooth projection operator, which is essential to project the continuous-valued worker parameter tuned by SASOC algorithms onto the discrete set. We validated our algorithms on five real-life service systems and compared them with a state-of-the-art optimization tool-kit OptQuest. Being 25 times faster than OptQuest, our algorithms are particularly suitable for adaptive labor staffing. Also, we observe that our 
algorithms guarantee convergence and find better solutions than OptQuest in many cases.
\end{abstract}


\section{Introduction}
A \emph{Service System (SS)} is an organization composed
of \begin{inparaenum}[(i)] \item the resources that support, and \item
  the processes that drive service interactions so that the outcomes
  meet customer expectations
  (\citep{alter2008service,spohrer2007steps})\end{inparaenum}. This
paper focuses on SS in the data-center management domain, where
customers own data centers and other IT infrastructures supporting
their businesses.  Owing to size, complexity, and uniqueness of these
technology installations, the management responsibilities of the same
are outsourced to specialized service providers.  A \emph{delivery
  center} is a remotely located workplace from where the service
providers manage the data-centers.  Each \emph{service request (SR)}
that arrives at a delivery center requires a specific skill and is
supported by a \emph{service worker (SW)} with the corresponding
skill set.  The SWs work in shifts which are typically aligned to the
business hours of the supported customers.  Hence, a group of
customers supported by a group of SWs, along with the operational
model of how SRs are routed constitutes an SS in this paper.  A
delivery center may consist of many SS.

We consider the problem of adaptive labor staffing in the context of
service systems. The objective is to find the optimal staffing levels
in a SS for a given dispatching policy (i.e., a map from service requests
to service workers) while maintaining system steady-state and
compliance to aggregate service level agreement (SLA) constraints.
The staffing levels constitute the worker parameter that we optimize and specify the number of workers in each shift and of each skill level.
The SLA constraints specify the target resolution time and the
aggregate percentage for an SR originating from a particular customer
and with a specified priority level. For instance, a sample SLA
constraint could specify that $95\%$ of all SRs from customer $1$ with
`urgent' priority must be resolved within $4$ hours. While the need for
SLA constraints to be met is obvious, the requirement for having
queues holding unresolved SRs bounded is also necessary because SLA
attainments are calculated only for the work completed. The problem is
challenging because analytical modeling of SS operations is difficult
due to aggregate SLA constraints and also because the SS
characteristics such as work patterns, technologies, and customers
supported change frequently. An important aspect to consider in the design of the adaptive labor staffing algorithm is its computational efficiency, as an algorithm with low running time helps in making staffing changes on a shorter timescale, for instance, every week.

We formulate this problem as a constrained Markov cost process
that depends on the worker parameter. 
To have a sense of the search space size, an SS
consisting of $30$ SWs who work in $6$ shifts and $3$ distinct skill
levels corresponds to more than $2$ trillion configurations.  We
design a novel single-stage cost function for the constrained Markov
cost process that balances the conflicting objectives of
worker under-utilization and SLA under/over-achievement. SLA under-achievement implies violation of the SLA constraint. Whereas worker under-utilization clearly points to suboptimal staffing, SLA over-achievement 
points to `over-delivery' and hence is also suboptimal. The
performance objective is a long-run average of this single stage cost
function and the goal is to find the optimum steady state worker parameter (i.e., the one that
minimizes this objective) from a discrete high-dimensional parameter set. However, our problem setting also involves
constraints relating to queue stability and SLA compliance. Thus, the
optimum worker parameter is in fact a constrained minimum. Another difficulty
in finding the optimum (constrained) worker parameter is
that the single stage cost and constraint functions can be estimated only via
simulation. Hence, the need is for a simulation-optimization algorithm that
incrementally updates the worker parameter along a descent direction,
while adhering to a set of queue stability and SLA constraints.

In this paper, we develop two novel discrete parameter simulation-based
optimization
algorithms for solving the above problem. Henceforth, we shall refer to
these algorithms as \emph{SASOC (Staff Allocation using Stochastic
  Optimization with Constraints)} algorithms. 
  The core of each of the algorithms is a multi-timescale stochastic
approximation scheme that incorporates a random perturbation based algorithm
for `primal descent' and couples it with a `dual ascent' scheme for the
Lagrange multipliers. The first order algorithm \footnote{A part of this work appeared as a short paper in ICSOC 2011 \citep{prashanth2011icsoc}} 
proposes the simultaneous perturbation stochastic approximation (SPSA) based
technique for gradient estimation in the primal. We also develop a second
order (Newton) methods that
estimates the Hessian of the objective function using SPSA and leverages
Woodbury's identity to directly estimate the inverse of the Hessian.
Both the SASOC algorithms that we propose are online, incremental and easy
to implement. 
Further, all SASOC algorithms involve a certain generalized smooth projection operator, which is essential to project the continuous-valued worker parameter tuned by SASOC algorithms onto the discrete set. The smoothness is necessary to ensure that the underlying transition dynamics of the constrained Markov cost process is itself smooth (as a function of the continuous-valued parameter) - a critical requirement to prove the convergence of all SASOC algorithms.
We evaluate our algorithms on five real-life SS in the
      data-center management domain.  For each of the SS, we collect
      operational data on work arrival patterns, service times, and
      contractual SLAs and feed this data into the simulation model
      of \citep{banerjee2011simulation}.
From the simulation experiments, we observe that our
      algorithms show overall better performance in comparison with
      the state-of-the-art OptQuest optimization toolkit
      (\citep{laguna1998optimization}). Further, our algorithms are $25$
      times faster than OptQuest and have a significantly lower execution runtime.

\subsection{Contributions to theory and methodology}
Newton-based algorithms usually suffer from the problem
of high per-iterate computational requirement because of
the need to estimate the inverse of the Hessian matrix at
each update epoch. We propose, for the first time, a method
for directly updating the inverse Hessian in Newton-based
Simultaneous Perturbation Stochastic Approximation algorithms
based on incorporating the Woodbury identity. This is seen to
result in significant computational savings as the resulting
Newton algorithm shows fast convergence. Our algorithm is based
on a novel generalized projection scheme. Since our problem
setting is one of discrete constrained optimization, we first
transform the problem for purposes of proving convergence (using the
proposed generalized projection scheme) to a continuous constrained
optimization setting. Note that SPSA is primarily a continuous
optimization technique. Our main observation is that SPSA also serves as a
powerful method in the context of discrete optimization even when
inequality constraints are considered. We prove the convergence
of the proposed SPSA algorithms. In the context of discrete optimization problems (with or
without inequality constraints) based on simulation, ours
is the first work that develops Newton-based
search algorithms.

\subsection{Contributions to practice}
Optimizing staff allocation in the context of service systems is challenging and the problem is further complicated by SLA constraints which are aggregate in nature.  Our system model (constrained Markov cost process) incorporates non-stationary workload arrivals and service times whose distribution is fitted from historical data and follows a lognormal (and not exponential) distribution. We present novel simulation optimization algorithms based on simulataneous perturbation technique that solve this problem. The proposed algorithms include both first order as well as second order optimization schemes and attempt to find the optimal staffing levels working with simulated data. Further, the proposed schemes are guaranteed to work with any given dispatching policy. Both our algorithms are online, incremental and computationally efficient - characteristics that make them amenable for their use in real service systems, especially with shorter periodicity for staff changes. From the numerical experiments based on 
data from real-life service systems, we observe that our SASOC algorithms exhibited overall superior performance in comparison to the state-of-the-art simulation optimization toolkit OptQuest. The experiments were performed with two different dispatching policies and it was observed that in each case SASOC algorithms converged rapidly to solutions of good quality at lower computational overhead as compared to OptQuest.        
 
%%%%%%%%%%%%%%%%%%%%%%%%%%%%%%%%%%%%%%%%%%%%%%%%%%%%%%%%%%%%
\section{Related Work}
We now review literature in two different areas of related work: (1)
techniques pertaining to service systems analysis and (2) developments
in stochastic optimization approaches.

\subsection{Service Systems}
In \citep{vermaautomated}, a two-step mixed-integer program is
formulated for the problem of dispatching SRs within service systems.
While their goal is similar to ours, their formulation does not model
the stochastic variations in arrivals and service times. Further,
unlike our framework, the SLAs in their formulation are
not aggregated over a month long period.  In \citep{wasserkrug2008creating}, the authors propose
a scheme for shift-scheduling in the context of third-level IT support
systems.  Unlike this paper, they do not validate their method against
data from real-life third-level IT support.
In \citep{cezik2008staffing,bhulai2008simple}, simulation-optimization methods
are proposed for finding the optimal staffing in a multiskill call center,
where the constraints are on long-term SLA requirements. While the paper by
\citep{cezik2008staffing} proposes a cutting plane algorithm for solving an
integer program, \citep{bhulai2008simple} relies on obtaining a linear
programming solution. However, unlike SASOC algorithms, steady-state system
analysis is not performed there. Instead, they solve a sample problem and show
that the optimal solution of the sample problem converges to that of the exact
problem when the number of samples go to infinity. 
% In \citep{chan2008analysis}, the emergent behavior of a service system
% consisting of a large number of cells is studied by applying an agent
% based simulation method.  Each cell contains an analytical M/M/1 queue
% model. The simulation helps observe how cells die and neighborhood
% patterns emerge among cells. While \citep{chan2008analysis} exemplifies
% the human aspects of service systems which would be an important
% future direction for our work, it does not aim to propose a
% labor-optimization technique.
In \citep{robbins2008simulation}, usage
of a simulation based search method is proposed for finding the
optimal staffing levels in the context of a call-center domain.  They
evaluate the system given a staffing level with an analytical model,
which is possible in their simplified domain but would not be feasible
for service systems due to aggregate SLA constraints and dynamic
queues. An analysis of service systems using the ARENA
simulation tool is presented in \citep{brickner2010simulation}. Unlike
our model, the system there is not subjected to aggregate SLA
constraints and they do not consider preemption of
low priority SRs by higher priority SRs and assignment of higher
skilled SWs to growing queues of SRs requiring lower skill levels.
% Their model assumes a pool of dedicated SWs for each customer as well
% as a separate pool of shared SWs.
In \citep{banerjee2011simulation}, a
simulation framework for evaluating dispatching policies is
proposed. While we share their simulation model, the goal in this paper
is to develop simulation optimization methods for optimizing the worker parameter in a constrained setting.  In general, none of the above papers
propose an optimization algorithm that is geared for SS and that leverages simulation to adapt optimization search parameters, when both the objective and the constrained functions are suitable long-run averages.

In \citep{prashanth2011ss}, some algorithms based on the smoothed functional technique
for gradient estimation were proposed for the
problem of staffing optimization in service systems. The algorithms 
there used certain random perturbations based on Gaussian and Cauchy density
functions to estimate the gradient of the Lagrangian. While we use random
perturbations using i.i.d.,~symmetric, $\pm 1$-valued, Bernoulli random variables, the computational
cost involved in our algorithms is significantly low when compared to \citep{prashanth2011ss}
because generating Bernoulli distributed random variables is significantly
less expensive than generating Gaussian or
Cauchy random variates. Further,
we also propose second-order Newton based methods, which are more robust than
the first order methods in the aforementioned reference. We compare our
proposed algorithms with the ones from \citep{prashanth2011ss} in the numerical
experiments.

\subsection{Stochastic Optimization}
SPSA (\citep{spall92multivariate}) is a popular and
highly efficient simulation based local optimization scheme for
gradient estimation. SPSA has the critical
advantage that it needs only two samples of the objective function to estimate its gradient for
any $N$-dimensional parameter. In \citep{spall1997one}, a
one-simulation variant of SPSA was proposed. However, the algorithm in
\citep{spall1997one} was not found to work as well in practice as its
two simulation counterpart.  Usage of deterministic perturbations
instead of randomized was proposed in \citep{bhatnagar2003two}. The
deterministic perturbations there were based either on lexicographic
or Hadamard matrix generated sequences and were found to perform better
than their randomized perturbation counterparts. Another approach that is seen to improve the
performance of gradient SPSA is to use a chaotic nonlinear random
number generator, see \citep{bhatnagar2003multiscale}. A Newton based SPSA algorithm that needs four system simulations with Bernoulli random perturbations was proposed in \citep{spall2000adaptive}. In \citep{bhatnagar2005adaptive}, three other SPSA based estimates of the Hessian that require three, two and one system simulations, respectively, were proposed. In \citep{bhatnagar2007adaptive}, certain smoothed functional (SF) Newton algorithms that incorporate Gaussian-based perturbations were proposed.
% The algorithms of \citep{spall2000adaptive}, \citep{bhatnagar2005adaptive} and \citep{bhatnagar2007adaptive} were for unconstrained optimization.
In \citep{shalabh2011stochastic} continuous optimization techniques such as SPSA and SF, have been adapted to a setting of discrete parameter optimization. Two simulation based optimization algorithms that involve randomized projections have been proposed there for an unconstrained setting.  
 In \citep{shalabh2011constrained}, several
simulation based algorithms for constrained optimization have been
proposed. Two of the algorithms proposed there use SPSA for estimating
the gradient, after applying the Lagrange relaxation procedure to the
constrained optimization problem, while the other two incorporate SF approximation.
% Constrained optimization in the
% context of Markov decision processes has been considered, for
% instance, in \citep{borkar2005actor,bhatnagar2010actor}. The algorithm
% proposed in \citep{borkar2005actor} is a three time-scale stochastic
% approximation scheme that incorporates an actor-critic algorithm for
% primal descent and performs dual ascent on Lagrange
% multipliers. However, it assumes full state representation for the
% underlying MDP. The algorithm proposed in
% \citep{bhatnagar2010actor} combines the ideas of multi-timescale
% stochastic approximation and reinforcement learning with function
% approximation to develop a simulation based online algorithm for a constrained control problem, involving function approximation.
For a detailed survey of gradient estimation techniques in the context of simulation optimization, the reader is referred to \citep{book}.

Our SASOC algorithms differ from the stochastic optimization
approaches outlined above in various ways. Many algorithms, for
instance those proposed in
\citep{spall2000adaptive,bhatnagar2005adaptive,bhatnagar2007adaptive},
are for unconstrained optimization and in a continuous optimization setting.
However, our staff optimization
problem is for a discrete worker parameter and requires SLA and queue stability constraints to be
satisfied in the long-run-average sense.
While the algorithms of \citep{shalabh2011constrained} have
been developed for constrained optimization in the case of a continuously-valued parameter, our SASOC
algorithms optimize a discrete parameter. Further, unlike
\citep{shalabh2011constrained} where an explicit inversion of the Hessian at each
update step was advocated, we incorporate the Woodbury's identity to obtain a
novel update step for the inverse of the Hessian in our algorithm SASOC-W.
Unlike \citep{shalabh2011stochastic} where fully randomized projections were
used, we incorporate a generalized projection operator that is continuously
differentiable in the parameter and works as a deterministic operator over a
large portion of the search space and incorporates randomization over a small
portion. This helps in bringing down the computational requirement as a
deterministic projection scheme requires less computation than a fully
randomized one.
To the best of our knowledge, we are the first to present adaptations of Newton-based search approaches for constrained discrete optimization problem. 

%%%%%%%%%%%%%%%%%%%%%%%%%%%%%%%%%%%%%%%%%%%%%%%%%%%%%%%%%%%%
% We now outline the contributions of this paper.
% \subsection{Our contributions}
The rest of the paper is organized as follows: First, we present the detailed problem formulation.
Second, we introduce our solution methodology and present 
simultaneous perturbation based SASOC algorithms for adaptive labor staffing. 
Third, we provide an outline of the convergence proof
and state the main results.\footnote{The detailed proofs have been provided
for review in a separate document attached to the paper.}
Fourth, we discuss the
implementation of our algorithms as well as the OptQuest algorithm and
present the performance simulation results. Finally, we provide the concluding remarks and discuss interesting future research directions.  

%%%%%%%%%%%%%%%%%%%%%%%%%%%%%%%%%%%%%%%%%%%%%%%%%%%%%%%%%%%%
\section{Problem Formulation}
\label{sec:formulation}
A service system is characterized by the following entities.
\begin{itemize}[$\bullet$]
    \item A finite set of customers, denoted by $\C$, supported by the service
system.
    \item A finite set of shifts, denoted by $\A$, across which the service
workers are distributed.
    \item A finite set of skill or complexity levels, denoted by $\B$.
    \item A finite set of priority levels, denoted by the set $P$.
    \item A finite set of time intervals, denoted by $\I$, where during each
interval the arrivals stay stationary, with the number of arrivals following a
Poisson distribution whose rate parameter is given by the function $\alpha$
described next. 
    \item Arrival rates specified by the mapping $\sigma: \C \times \I
\rightarrow \R$. We assume that each of the SR arrival processes from the
various customers $C_i$ are independent and Poisson distributed with
$\alpha(C_i,I_j)$ specifying the rate parameter. Owing to the finite-buffer
nature of the system, we assume that the number of arrivals during any
interval ($\in \I$) is upper-bounded by a sufficiently large constant.
    \item Service time distributions characterized by the mapping $\tau: P \times \B \rightarrow (r_1, r_2), r_i \in \R, i=1,2$. Here $r_1$ represents the mean and $r_2$ the standard deviation of a truncated lognormal distributed random variable corresponding to a particular priority-complexity pair. In other words, if $M$ is a random variable following a normal distribution with mean $r_1$ and standard deviation $r_2$, then the truncated lognormal random variable is $e^M \wedge \top$, where $\top$ is a truncation constant that is chosen to be large in practice.     
    \item SLA constraints, given by the mapping $\gamma: \C \times P
\rightarrow (r_1, r_2), r_i \in \R, i = 1,2$. Here $\gamma(C_i,P_j) = (r_1,r_2)$
implies that the SLA target for SRs from customer $C_i$ and with priority $P_j$ is
$(r_1,r_2)$, with $r_1$ specifying the SLA percentage target and $r_2$ the
resolution time target (in hours). For instance, $\gamma(C_1,P_1) = (95,4)$
translates to the requirement that at least $95\%$ of the SRs from customer
$C_1$ with priority level $P_1$ should be closed within $4$ hours. Note that the
SLAs are computed at the end of each month and hence the aggregate SLA targets
are applicable to all SRs that are closed within the month under consideration.
Henceforth, we shall adopt the notation $\gamma_{i,j}$ to denote
$\gamma(C_i,P_j)$. 
\end{itemize}
Note that each arriving SR has a customer identifier ($\in \C$), a
priority identifier ($\in P$) and a complexity identifier ($\in \B$), whereas
any SW works in a particular shift
($\in \A$) and possesses a skill level ($\in \B$). In other words, each customer
can issue multiple SRs with their respective SLA targets and the SWs with the
right skill level and relevant shift have to pull these SRs from the complexity
queues and close them within the deadline specified by the SLA. The set $\I$ and
the mapping $\alpha$ allow us to model the variations in arrival rates better
than in a setting where the arrivals are assumed to be Poisson-distributed for
the entire period. 
Further, the time taken by an SW to complete an SR is stochastic and follows a
lognormal distribution, where the parameters of the distribution are
learned by conducting time and motion exercises described in
\citep{banerjee2011simulation}.

\begin{table}
\centering
\caption{Sample workers, utilizations and SLA targets}
\label{table:workers}
\begin{tabular}{c}
\subfigure[Workers $\theta_{i}$]{
\label{servicesystems:workers}
\begin{tabular}{|l|rrr|}
\hline
\textbf{} & \multicolumn{ 3}{c|}{\textbf{Skill levels}} \\\hline
\textbf{Shift} & \multicolumn{1}{l}{High} & \multicolumn{1}{l}{Med} & \multicolumn{1}{l|}{Low} \\ \hline
S1 & 1 & 3 & 7 \\
S2 & 0 & 5 & 2 \\
S3 & 3 & 1 & 2 \\ \hline
\end{tabular}} 
\vspace{3ex}

\\
\subfigure[Utilizations $u_{i,j}$]{
\label{servicesystems:utilizations}
\begin{tabular}{|l|rrr|}
\hline
\textbf{} & \multicolumn{ 3}{c|}{\textbf{Skill levels}} \\\hline
\textbf{Shift} & \multicolumn{1}{l}{High} & \multicolumn{1}{l}{Med} & \multicolumn{1}{l|}{Low} \\ \hline
S1 & 67\% & 34\% & 26\% \\
S2 & 45\% & 55\% & 39\% \\
S3 & 23\% & 77\% & 62\% \\ \hline
\end{tabular}}

\vspace{3ex}

\\
\subfigure[SLA targets $\gamma_{i,j}$]{
\label{ss:targetsla}
\begin{tabular}{|l|rr|}
\hline
\textbf{} & \multicolumn{ 2}{c|}{\textbf{Customers}} \\\hline
\textbf{Priority} & \multicolumn{1}{l}{Bossy Corp} & \multicolumn{1}{l|}{Cool Inc} \\ \hline
$P_1$ & 95\%4h & 89\%5h  \\
$P_2$ & 95\%8h & 98\%12h  \\
$P_3$ & 100\%24h & 95\%48h \\
$P_4$ & 100\%18h & 95\%144h  \\ \hline
\end{tabular}}
\end{tabular}
\end{table}

Table \ref{servicesystems:workers} illustrates a
simple SS configuration, specifying the staffing levels across shifts
and skill levels. This essentially constitutes the worker parameter
that we optimize. In this example, $\A = \{$S1, S2,
S3$\}$ and $\B=\{$high, medium, low$\}$. Tables
\ref{servicesystems:utilizations} and
\ref{ss:targetsla} provide sample
utilizations and SLA targets on a SS with three shifts, two
customers and four priority levels.

Figure \ref{fig_ss} shows the main components of the SS. The SRs
arrive from multiple customers and the arrival rate is specific to the
hour of week, i.e., within each hour of week, and for each
customer-priority pair, the arrivals follow a Poisson
distribution.  The parameters of this distribution are learned from
historical data over a period of at least $6$ months.  Once the SR
arrives, it is queued up in a matching complexity queue by the queue
manager and the dispatcher would then assign it to an SW based on the
dispatching policy.  For instance, in the PRIO-PULL policy, SRs are
queued in the complexity queues based directly on the priority
assigned to them by the customers. On the other hand, in the EDF
policy, the time left to SLA target deadline is used to assign the SRs
to the SWs i.e., the SW works on the SR that has the earliest
deadline. Note that we have a finite buffer system, i.e., the number of SRs in each of the 
complexity queues is upper-bounded by a sufficiently large constant. Any arriving SR that finds the corresponding complexity queue full will depart the system.

A SW works in exactly one shift (working days and times) and a SS may operate in multiple shifts.
We say that a particular configuration of workers across shifts and skill levels is feasible if 
\begin{inparaenum}[(a)]
\item the SLA constraints are met and
\item the complexity queues do not become unbounded
\end{inparaenum}
when using this configuration.
While the need for (a) is obvious, the requirement for having bounded complexity queues is also necessary. This is because SLA attainments are calculated only for work completed and not for work waiting for completion in the complexity queues. For instance, say in a given month, $100$ SRs arrive at various times from a customer to a SS and only $50$ of them are completed within the target completion time stipulated by the SLA constraints. The remaining $50$ SRs are still in progress without a known completion time and hence do not have an impact on the SLA attainment measures. Thus, a healthy SLA attainment alone is insufficient and 
the bound on the growth of complexity queues fills the gap.

%%%%%%%%%%%%%%%%%%%%%%%%%%%%%%%%%%%%%%%%%%%%%%%%%%%%%%%%%%%%
\subsection{Constrained parameterized Markov Cost Process}
\label{sec:constrained-markov}

\begin{figure}
\begin{minipage}[c][\textheight]{\textwidth} 
    \centering
    \includegraphics[height=4in]{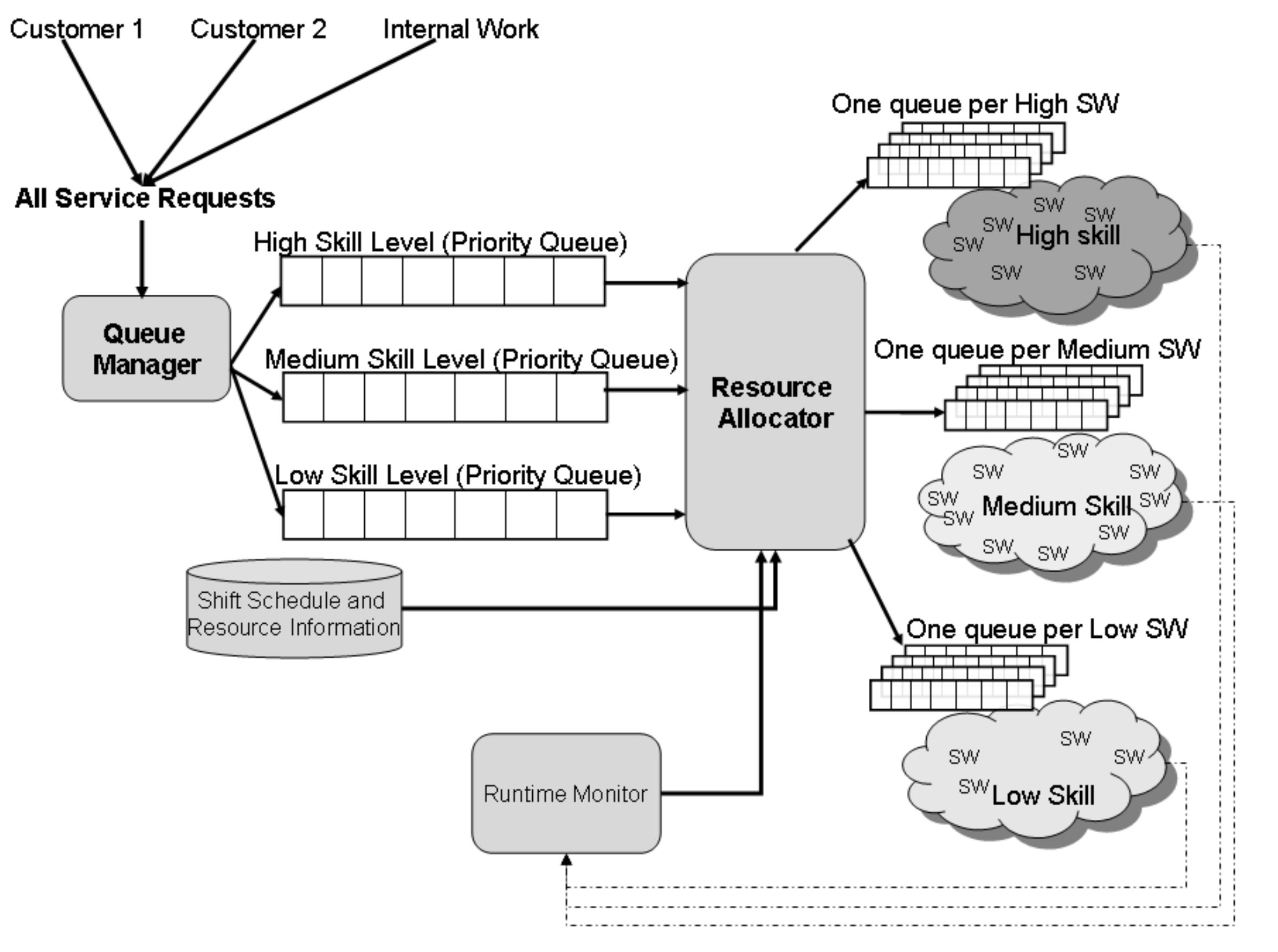}
    \caption{Operational model of an SS}
    \label{fig_ss}
\end{minipage}
\end{figure}

\begin{figure}
\begin{minipage}[c][\textheight]{\textwidth} 
\centering
\begin{tikzpicture}
\draw[very thick] (-0.2,0) -- (6.2,0);
\draw (0, 0.3) -- (0, -0.2) node [below=1pt] (n) {$n\T$};
% \draw (1, 0.2) -- (1, -0.2) node [below] (n1) {$n + 1$};
\draw (6, 0.3) -- (6, -0.2) node [below] (n1) {$(n+1)\T$};
\node [below=4pt of n] (hatt) {$X_n$};
\node [left=8pt of n] (alloc) {Instant};
\draw[>=latex',->] (alloc) -- (n);
\node [left=8pt of hatt] (alloc) {State};
\draw[>=latex',->] (alloc) -- (hatt);

\node [below=4pt of n1] (bart) {$X_{n+1}$};
% \node [below=8pt of bart] (complete) {Completion by $i$};
% \draw[>=latex',->] (complete) -- (bart);

\draw[>=latex',->] (0,0.25) -- node[above] {Simulate($\theta(n),\T$)} (6,0.25);
\end{tikzpicture}
\caption{A portion of the time-line illustrating the process}
\label{fig:timeline}
\end{minipage}
\end{figure}
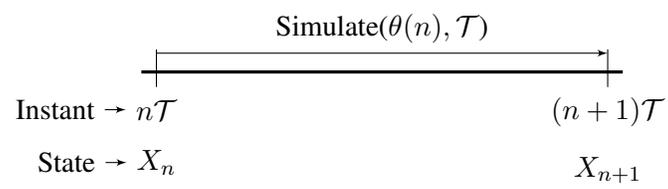

We consider the setting of a constrained parameterized Markov cost process that we describe
in detail below
\footnote{A similar framework is
considered, for instance, in \citep{marbach2001simulation,prashanth2011ss}. However, the setting
considered in \citep{marbach2001simulation} is unconstrained and the parameter is continuous-valued. Our formulation, though similar to that in \citep{prashanth2011ss}, is simpler as it does not involve hidden state components.}. 
Our setting, however, involves a discrete-time, continuous-space Markov process
represented by $\{X_n(\theta), n \ge0\}$. We describe
$X_n$ more clearly in Section \ref{sect:scc}. The transition
probabilities of this process depend on the worker parameter $\theta  =
(\theta_{1},\ldots,\theta_N)^T \in \D$,
where $N = |A| \times |B|$. In the above, $\theta_{i}$ indicates the number of
service workers whose skill level is $(i-1) \% |B|$ and whose shift index is
$(i-1) / |B|$. As an example, the worker parameter for the setting in Table
\ref{servicesystems:workers} is $\theta = (\theta_1, \ldots, \theta_9)^T$ $=
(1,3,7,0,5,2,3,1,2)^T$.
The parameter vector $\theta$ takes values in the set $\D$, where ${\displaystyle \D
\stackrel{\triangle}{=}
 \{0,1, \ldots, W_{\max}\}^N }$. Here
$W_{\max}$ serves as an upper bound for the number of workers in any shift and of any skill level. Note that one can enumerate all the points in $\D$ as $\D = \{D^1, D^2, \ldots, D^p\}$ for some $p >1$.

As illustrated in Figure \ref{fig:timeline}, the system stochastically transitions
from one state to another, while incurring a state-dependent cost. In addition,
there are state-dependent single-stage (constraint) functions described via 
$g_{i,j}(X_n), h(X_n), i = 1, \dots, |C|, j = 1, \dots, |P|$. These shall
correspond to the SLA and queue stability constraints. The state together with
the cost and constraint functions constitutes the constrained Markov cost
process. The $n$th system transition of this underlying process involves a
simulation of the service system for a fixed period $\T$ with the current worker
parameter $\theta(n)$. However, 
arrivals are stopped after time $\T$ and the service system is simulated until
the complexity queues are empty. In our experiments, $\T = 10$,
i.e., we simulate the service system for a period of ten months with the
staffing levels specified by $\theta(n)$. 
Also, note that this is a continuously running
simulation where, at discrete time instants $n\T$,
we update the worker parameter $\theta(n)$ and the simulation output causes a
probabilistic transition from the current state $X_n$ to the next state
$X_{n+1}$, while incurring a single-stage cost $c(X_n)$.  The precise
definitions of the state, the cost and the constraints functions are given in
Section \ref{sect:scc}. By an abuse of notation, we refer to the state at
instant $n\T$ as $X_n$.

\subsection{The Objective}
\label{sect:obj}
We use the long-run average cost as the performance objective in our setting. Thus, we are interested in optimizing the steady-state system performance. The optimization problem is the following:
\begin{equation}
\label{eqn:c-mp}
\begin{array}{l}
\textrm{Find } \min\limits_{\theta} J(\theta) \stackrel{\triangle}{=} \lim\limits_{n \rightarrow \infty}\frac{1}{n} \sum\limits_{m=0}^{n-1} c(X_m)\\
\textrm{subject to}\\
G_{i,j}(\theta) \stackrel{\triangle}{=} \lim\limits_{n \rightarrow \infty}\frac{1}{n} \sum\limits_{m=0}^{n-1} g_{i,j}(X_m) \le 0, \\\qquad\qquad\qquad\quad\forall i=1,\ldots,|C|, j=1,\ldots,|P|,\\
H(\theta) \stackrel{\triangle}{=} \lim\limits_{n \rightarrow \infty}\frac{1}{n} \sum\limits_{m=0}^{n-1} h(X_m) \le 0.
%g_{i,j}(X_m) =  \le 0 \quad\forall i=1,\ldots,|C|, j=1,\ldots,|P|
\end{array}
\end{equation}
We assume below that the Markov process $\{X_n\}$ under any parameter $\theta$
is ergodic. In such a case, the limits in \eqref{eqn:c-mp} are well-defined. If
this is not the case, one may replace the ``lim'' with ``limsup'' in the
definitions of $J(\theta), G_{i,j}(\theta)$ and $H(\theta)$ in
\eqref{eqn:c-mp}. 
Given the above constrained Markov cost process formulation, the optimization problem \eqref{eqn:c-mp} essentially stipulates that the optimal worker parameter $\theta^*$ should minimize the long-run average cost objective $J(\cdot)$ while maintaining queue stability in steady-state  (i.e., the long-run average of  $h(X_n)$ should not be above zero) and adhering to contractual SLAs, i.e., that the long-run average of $g_{i,j}(X_n)$ should not be above zero, for any feasible $(i,j)$-tuple.

The SASOC algorithms that we design subsequently (see Section
\ref{sec:algorithm}) use the cost $c(X_n)$ and constraint functions
$g_{i,j}(X_n), h(X_n)$ to tune the worker parameter $\theta(n)$ at instant $n\T$
and the system simulation would now continue with the updated worker parameter. 
While it is desirable to find the optimum  $\theta^* \in S$, i.e.,
\begin{align*}
\label{eq:optimal-parameter-set}
\theta^* = \mathop{{\rm argmin}} \bigg\{ J(\theta) \text{ s.t. } \theta \in \D, G_{i,j}(\theta) \le 0, i=1,\ldots,|C|,j=1,\ldots,|P|,H(\theta) \le 0\bigg\},
\end{align*}
it is in general very difficult to achieve a global minimum. We apply the Lagrange relaxation procedure to the above problem and then provide SPSA based algorithms - both first as well as second order, for finding a locally optimum parameter $\theta^*$.
We now describe in detail the state, single-stage cost and constraint functions that we adopt for the constrained Markov cost process formulated for optimizing the staffing in the context of service systems.

\subsection{State, Cost and Constraints}
\label{sect:scc}
The state $X_n$ at instant $n$ is the vector of the length of waiting SR
queues corresponding to each skill level, the current  utilization of workers
for each shift and skill level, and the current SLA attainments for each
customer and SR priority. Thus,
\begin{align}
X_{n}  = ( \N(n), u(n), \gamma'(n), q(n)),
\end{align}
where,
\begin{itemize}[$\bullet$]
%     \item $\theta(n)$ is the worker parameter that gives the current staffing levels across shifts and skill levels and is defined in \eqref{eq:theta}.
    \item $\N(n) = (\N_1(n), \ldots, \N_{|B|}(n))^T$, with  $\N_i(n)$ being the number of SRs in 
the system queue corresponding to skill level $i \in \B$. As all the complexity queues are of 
finite size, we have $\N_i(n) \le \varsigma , i=1,\ldots,|B|$, where $\varsigma >0$ 
is a sufficiently large constant.
    \item The utilization vector $u(n) = (u_{1,1}(n),\ldots,u_{|A|,|B|}(n))$,
with each  $u_{i,j}(n) \in [0,1]$ being the average utilization of the workers in
shift $i$ and skill level $j$, at instant $n$.
    \item The SLA attainment vector $\gamma'(n) = (\gamma'_{1,1}(n),\ldots,\gamma'_{|C|,|P|}(n))$, 
with $\gamma'_{i,j}(n) \in [0,1]$ being the SLA attainment for customer $i$ and priority $j$, at instant $n$.
    \item $q(n)$ is an indicator variable that denotes the queue feasibility status of the system at instant $n$. In other words,  $q(n)$ is $0$ if the growth rate of the SR queues (for each complexity) is beyond a threshold and is $1$ otherwise. We need $q(n)$ to ensure system steady-state which is independent of SLA attainments because the latter are computed only
on the SRs that were completed and not on those queued up in the system.
\end{itemize}
Let $S$ denote the state space. We observe that $S$ is a compact set. This is
because each of the state components in $X_n$ take values in sets that are closed and bounded. In
particular, each element of $u(n)$, $\gamma'(n)$ takes values in $[0,1]$ and $0
\le q(n) \le 1$, respectively. The system SR queues $\N$ are also of finite length and hence,
$X_n$ is bounded. 

Considering that the queue lengths, utilizations and SLA attainments at instant
$n+1$ depend only on the state $X_n$ at instant $n$, we
observe that $\{X_n(\theta), n \ge 0\}$ is a constrained 
Markov cost process for any given (fixed) parameter $\theta$.
% Further, note that the utilizations and SLA attainments are real numbers and hence we have a Markov chain in continuous space. Further, owing to $0 \le u_{i,j}(n), \gamma'_{i,j}(n) \le 1$, the state space is bounded.
We now describe in detail the single-stage cost function, whose long-run average sum we try to optimize in \eqref{eqn:c-mp}.
We let the cost function $c(X_n)$ have the form:
\begin{equation}
\label{eqn:singlestagecost}
\begin{array}{l}
c(X_n)  =    r \times \left( 1 -  \sum_{i=1}^{|A|}\sum_{j=1}^{|B|} \alpha_{i,j} \times u_{i,j}(n) \right) + s \times \left( \dfrac{\sum_{i=1}^{|C|}\sum_{j=1}^{|P|} \left | \gamma'_{i,j}(n) - \gamma_{i,j} \right |}{|C|\times|P|}\right),
\end{array}
\end{equation}
where $r,s \ge 0$ and $r + s =1$. Further, $0 \le \gamma_{i, j} \le 1$ denotes the contractual SLA for customer $i$ and priority $j$. The single-stage cost function here is a linear function of the state and remains bounded. In fact, from \eqref{eqn:singlestagecost}, we observe that $0 \le c(X_n) \le 1$. This is because $u_{i,j}(n), \gamma_{i,j}, \gamma'_{i,j}(n) \in [0,1]$ and each component in \eqref{eqn:singlestagecost} is upper-bounded by $1$.

The cost function is designed to balance between two conflicting objectives of maximizing the utilization of workers and meeting the SLA requirements simultaneously.
By the first component in \eqref{eqn:singlestagecost}, we seek to minimize the under-utilization of workers as it is more fine-grained and hence, allows tighter minimization in comparison to minimizing just the sum of workers across shifts and skill levels. The second component in \eqref{eqn:singlestagecost} represents the over/under-achievement of SLAs, which is the distance between attained and the contractual SLAs. While the need for meeting the target SLAs motivates the under-achievement part in the second component, it is also necessary to minimize over-achievement of SLAs. This is because an over-achieved SLA, for instance meeting $100\%$ instead of the target of $95\%$ for a particular customer, while being desirable for the customer,
requires more time and effort from some of the workers and does not bring in additional rewards.

Note that the first term in (\ref{eqn:singlestagecost}) uses a weighted sum of utilizations over workers from each shift and across each skill level.  Further, the weights $\alpha_{i,j}$ are fixed and not time-varying. Using historical data on SR arrivals, the percentage of workload arriving in each shift and for each skill level is obtained. These percentages decide the weights $\alpha_{i,j}$ used in \eqref{eqn:singlestagecost}, that in turn satisfy
\[0 \le \alpha_{i,j} \le 1, \textrm{ and }\sum_{i=1}^{|A|}\sum_{j=1}^{|B|} \alpha_{i,j} = 1,\]
for $i=1,2,\ldots,|A|,$ and $j=1,2,\ldots,|B|$. This prioritization of workers helps in optimizing the worker set based on the given workload. For instance, if $70\%$ of the SRs requiring low skill worker attention arrive in shift $1$, then one may set $\alpha_{1,0} = 0.7$, in the cost function \eqref{eqn:singlestagecost}, where $0$ denotes the low skill level index.

The single-stage constraint functions $g_{i,j}(\cdot), h(\cdot), i = 1, \dots, |C|, j = 1, \dots, |P|,$ are given by:
\begin{align}
g_{i,j}(X_n) & = \gamma_{i,j} - \gamma'_{i,j}(n), \forall i=1,\ldots,|C|, j=1,\ldots,|P|,\label{eqn:sla-constraints}\\[1ex]
h(X_n) & = 1 - q(n).\label{eqn:feasibility-constraint}
\end{align}
Here (\ref{eqn:sla-constraints}) specifies that the attained SLA levels should be equal to or above the contractual SLA levels for each customer-priority tuple. Further, (\ref{eqn:feasibility-constraint}) ensures that the SR queues for each complexity in the system stay bounded. In the constrained optimization problem formulated below, we attempt to satisfy these constraints in the long-run average sense (see (\ref{eqn:c-mp})).

The SASOC algorithms treat the parameter as continuous-valued and tune it accordingly. Let us denote this continuous version of the worker parameter by $\bar\theta = (\bar\theta_1,\ldots,\bar\theta_N)$. Note that $\bar\theta_i \in [0, W_{\max}], i=1,2,\ldots,N$. We now design a smooth projection operator $\Gamma$ that projects $\bar\theta$ on to the discrete space $\D$ so that the same can be used for performing the simulation of the service system. We call the $\Gamma$-operator as a generalized projection scheme as it lies in between a fully deterministic projection scheme based on mere rounding off and a completely randomized scheme, whereby depending on the value of $\bar\theta_j$ (for any $j=1,\ldots,N$) one can find points $\D^k$ and $\D^{k+1}$ with $\D^k < \D^{k + 1}$, $\D^k, \D^{k + 1} \in \D$ such that $\D^k$ and $\D^{k + 1}$ are the immediate neighbours of $\bar\theta_j$ in the set $\D$. Then, one sets the corresponding discrete parameter as
\begin{equation}
 \theta_j =
  \begin{cases}
   \D^{k + 1} &  \text{w.p. } \dfrac{\bar\theta_j - \D^k }{\D^{k + 1} - \D^k}, \\
   \D^k &  \text{w.p. } \dfrac{\D^{k + 1} - \bar\theta_j }{\D^{k + 1} - \D^k},
  \end{cases}
\end{equation}
where, w.p.~ stands for `with probability'.

\subsection{A Generalized Projection Operator}
\label{sec:gammaproj}
For any $\bar\theta = (\bar\theta_1,\ldots,\bar\theta_N)$ with $\bar\theta_j \in
[0, W_{\max}], j = 1,2,\ldots,N$, we define a projection operator
$\Gamma(\bar\theta) = (\Gamma_1(\bar\theta_1),\ldots,\Gamma_N(\bar\theta_N)) \in
\D$ which projects any $\bar\theta$ onto the discrete set $\D$ as follows:

For convenience, lets enumerate the elements of $\D$ as
$\D = \{\D^1,\D^2,\ldots,\D^p\}$ for some $p>1$.
Let $\zeta >0$ be a fixed real number and $\bar\theta_i$ be such that $\D^j \le
\bar\theta_i \le \D^{j+1}, \D^j < \D^{j+1}$ for some $\D^j, \D^{j+1} \in \D$.
Let us consider an interval of length $2 \zeta$ around the midpoint of $[\D^j,
\D^{j+1}]$ and denote it as $[\tilde \D_1, \tilde \D_2]$, where $\tilde \D_1 =
\frac{\D^j+\D^{j+1}}{2} - \zeta$ and $\tilde \D_2 = \frac{\D^j+\D^{j+1}}{2} +
\zeta$.
 Then, \\$\Gamma_i(\bar\theta_i)$ for $\theta_i \in [\D^j,\tilde\D_1] \cup
[\tilde\D_2,\D^{j+1}]$ is defined by
\begin{equation}
 \Gamma_i(\bar\theta_i) =
  \begin{cases}
   0 &  \text{if } \bar\theta_i < 0 \\
   \D^j & \text{if } \bar\theta_i \le \frac{\D^j+\D^{j+1}}{2} - \zeta  \\
   \D^{j+1} & \text{if } \bar\theta_i \ge \frac{\D^j+\D^{j+1}}{2} + \zeta \\
   W_{\max} & \text{if } \bar\theta_i \ge W_{\max}.
  \end{cases}
\end{equation}
Further, $\Gamma_i(\bar\theta_i)$ for $\bar{\theta}_i \in [\tilde\D_1,\tilde\D_2]$ is
given by
\begin{equation}
 \Gamma_i(\bar\theta_i) =
  \begin{cases}
   \D^j & \text{w.p. } f(\frac{\tilde\D_2 - \bar\theta_i}{2\zeta})   \\
   \D^{j+1} & \text{w.p. } 1- f(\frac{\tilde\D_2 - \bar\theta_i}{2\zeta}) \\
  \end{cases}
\end{equation}
In the above, $f$ is any continuously
differentiable function defined on $[0,1]$ such that $f(0)=0$ and $f(1)=1$. Note
that we deterministically project onto either $\D^j$ or $\D^{j+1}$ if
$\bar\theta_i$ is outside of the interval $[\tilde\D_1, \tilde\D_2]$. Further,
for $\bar\theta_i \in [\tilde\D_1, \tilde\D_2]$, we project randomly using a
smooth function $f$. It is necessary to have a smooth projection operator to
ensure convergence of our SASOC algorithms as opposed to a deterministic
projection operator that would project $\bar\theta_i \in
[\D^j,\frac{\D^j+\D^{j+1}}{2})$ to $\D^j$ and $\bar\theta_i \in
[\frac{\D^j+\D^{j+1}}{2}, \D^{j+1}]$ to $\D^{j+1}$. The problem with a
deterministic projection operator is that there is a jump at the midpoint of the interval and hence, when
extended for any $\theta$ in the convex hull $\bar\D$, the transition dynamics
of the process $\{X_n, n\ge0\}$  is not continuously
differentiable. A non-smooth projection operator makes the dynamics non-smooth at the boundary
points.

The SASOC algorithms that we present subsequently tune the worker parameter in
the convex hull of $\D$, denoted by $\bar\D$, a set that can be defined as
$\bar\D = [0, W_{\max}]^N$. This idea has been used in
\citep{shalabh2011stochastic} for an unconstrained discrete
optimization
problem. However, the projection operator used there was a fully randomized
operator. The generalized projection scheme that we incorporate has the
advantage that while it ensures that the transition dynamics of the parameter
extended Markov process is smooth (as desired), it requires a lower
computational effort because in a large portion of the parameter space (assuming
$\zeta$ is small), the projection operator is essentially deterministic.

We also require another projection operator $\bar\Gamma$ that projects any
$\theta \in \R^N$ onto the set $\bar\D$ and is defined as $\bar\Gamma(\theta) =
(\bar\Gamma_1(\theta_1), \ldots, \bar\Gamma_N(\theta_N))$, where
$\bar\Gamma_i(\theta_i) = \min(0, \max(\theta_i, W_{\max}))$, $i=1,\ldots,N$.
Thus, $\bar\Gamma(\cdot)$ keeps the parameter updates within the set $\bar\D$
and $\Gamma(\cdot)$ projects them to the discrete set $\D$. The projected
updates are then used as the parameter values for conducting the simulation of
the service system.

\subsection{Assumptions}
We now make the following standard assumptions: One amongst (A2) and (A2')
will be assumed for the algorithms that follow.

\begin{description}
\item[\textbf{(A1)}] The Markov process $\{X_n(\theta) , n \geq
0\}$ under a given dispatching policy and parameter $\theta$ is ergodic.
\item[\textbf{(A2)}] The single-stage cost functions $c(\cdot)$,
$g_{i,j}(\cdot)$ and $h(\cdot)$ are all continuous. The long-run average cost
$J(\cdot)$ and constraint functions $G_{i,j}(\cdot), H(\cdot)$ are twice
continuously differentiable with bounded third derivative.
\item[\textbf{(A2')}] The single-stage cost functions $c(\cdot)$,
$g_{i,j}(\cdot)$ and $h(\cdot)$ are all continuous. The long-run average cost
$J(\cdot)$  and constraint functions $G_{i,j}(\cdot), H(\cdot)$ are continuously
differentiable with bounded second derivative .
\item[\textbf{(A3)}] The step-sizes $\{a(n)\}$, $\{b(n)\}$ and $\{d(n)\}$
satisfy
\[
\begin{array}{l}
\sum_{n} a(n) = \sum_{n} b(n) = \sum_n d(n) =\infty; \sum_n (a^2(n) + b^2(n) +
d^2(n)) < \infty,\\[2ex]
 %b(n) = o(d(n)) \text{ and }  a(n) = o(b(n))
 \dfrac{b(n)}{d(n)}, \dfrac{a(n)}{b(n)} \rightarrow 0 \text{ as } n \rightarrow
\infty.
\end{array}
\]
\end{description}

Assumption (A1) ensures that the process $\{X_n\}$ is stable for any given
$\theta$ and ensures that the long-run averages of the single stage cost and
constraint functions in \eqref{eqn:c-mp} are well-defined.
As stated earlier, we require one of (A2) and (A2') for our various algorithms. More
specifically, (A2) will be assumed for Hessian based schemes, while (A2') will
be assumed for gradient approaches. (A2) and (A2') are technical requirements
needed to push through suitable Taylor's arguments in order to prove the
convergence of the algorithms. The first two conditions in (A3) are standard
requirements for step-size sequences and the last condition there ensures a separation of
time scales between the different recursions in SASOC algorithms discussed in
detail in Section \ref{sec:algorithm}.

\begin{remark}
    As seen before, the state space $S$ is compact and ergodicity of the
underlying Markov process $\{X_n\}$ will follow if one ensures that there
is at least one worker for each complexity class.
\end{remark}
%%%%%%%%%%%%%%%%%%%%%%%%%%%%%%%%%%%%%%%%%%%%%%%%%%%%%%%%%%%%
\section{Our algorithms}
\label{sec:algorithm}
The constrained long-run average cost optimization problem
(\ref{eqn:c-mp}) can be expressed using the standard Lagrange
multiplier theory as an unconstrained optimization problem given
below.
\begin{eqnarray}
\nonumber
\max_{\lambda} \min_{\theta} L(\theta, \lambda) &\stackrel{\triangle}{=}& \lim_{n \rightarrow \infty}\frac{1}{n} \sum_{m=0}^{n-1} E\left \{ c(X_m) + \sum\limits_{i = 1}^{|C|} \sum\limits_{j = 1}^{|P|}\lambda_{i,j} g_{i,j}(X_m) + \lambda_f h(X_m) \right \}\\
\label{eqn:Lagrangian}
&=& J(\theta) + \sum\limits_{i = 1}^{|C|} \sum\limits_{j = 1}^{|P|}\lambda_{i,j} G_{i,j}(\theta)
+ \lambda_f H(\theta),
\end{eqnarray}
where $\lambda_{i,j} \ge 0, \quad \forall i=1,\ldots,|C|,
j=1,\ldots,|P|$ represent the Lagrange multipliers corresponding to
constraints $g_{i,j}(\cdot)$ and $\lambda_f$ represents the Lagrange
multiplier for the constraint $h(\cdot)$, in the optimization problem
(\ref{eqn:c-mp}). Also, $\lambda = (\lambda_{i,j}, \lambda_f,
i=1,\ldots,|C|, j=1,\ldots,|P|)^T$. The function $L(\theta, \lambda)$
is commonly referred to as the Lagrangian. An optimal $\left (
\theta^*, \lambda^* \right )$ is a saddle point for the Lagrangian,
i.e., $L(\theta, \lambda^*) \ge L(\theta^*, \lambda^*) \ge L(\theta^*,
\lambda)$, $\forall\theta$, $\forall\lambda$. Thus, it is necessary to design an algorithm which descends
in $\theta$ and ascends in $\lambda$ in order to find the optimum point. The
simplest iterative procedure for this purpose would use the gradients
of the Lagrangian with respect to $\theta$ and $\lambda$ to descend
and ascend respectively. However, for the given system, the computation
of gradient with respect to $\theta$ would be intractable due to lack
of a closed form expression of the Lagrangian. Thus, a simulation
based algorithm is required. The above explanation suggests that an
algorithm for computing an optimal $\left ( \theta^*, \lambda^* \right
)$ would need three stages in each of its iterations.

\begin{enumerate}
 \item The inner-most stage which performs one or more simulations
   over several time steps and aggregates data, i.e., does the averaging of the single-stage cost and constraint functions $c(\cdot), g_{i,j}(\cdot)$ and $h(\cdot)$ for any given $\theta$ and $\lambda$ updates.
 \item The next outer stage which estimates the gradient of the Lagrangian along $\theta$ and updates $\theta$ along a descent direction. This stage would perform several iterations for a given $\lambda$ and find a good estimate of $\theta$; and
 \item The outer-most stage which updates the Lagrange multipliers $\lambda$ along an ascent direction, using the converged values of the inner two loops.
\end{enumerate}

The above three steps will have to be performed iteratively till the
solution converges to a saddle point described previously.
Note that the loops are nested in the sense that the loop in iteration (1) would be a sub-loop for iteration (2). Likewise, iteration (2) would be a sub-loop for iteration (3). Thus, in between two successive updates of an outer loop (iterations (2) or (3)), one would potentially have to wait for a long time for convergence of the inner loop procedure (iteration (1) or iterations (1) and (2), respectively). This problem gets addressed by using
simultaneous updates to all three stages in a stochastic recursive
scheme but with different step-size schedules, the outer-most having
the smallest while the inner-most having the largest of step-sizes. The
resulting scheme is a multiple time-scale stochastic approximation
algorithm \citep[Chapter 6]{borkar2008stochastic}.

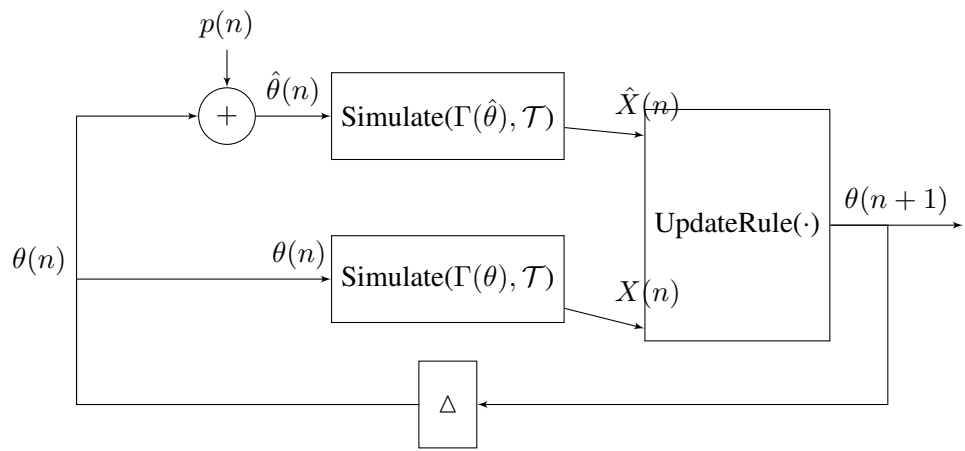
\begin{figure}
\begin{minipage}[c][\textheight]{\textwidth}
\centering
\tikzstyle{block} = [draw, fill=white, rectangle,
    minimum height=3em, minimum width=6em]
\tikzstyle{sum} = [draw, fill=white, circle, node distance=1cm]
\tikzstyle{input} = [coordinate]
\tikzstyle{output} = [coordinate]
\tikzstyle{pinstyle} = [pin edge={to-,thin,black}]
\scalebox{1.0}{\begin{tikzpicture}[auto, node distance=2cm,>=latex']
    % We start by placing the blocks
    \node (theta) {};
    \node [sum, above right=0.8cm of theta, xshift=1cm] (perturb) {$+$};
    \node [above=0.5cm of perturb] (noise) {$p(n)$};
    \node [block, right=1cm of perturb] (psim) {Simulate($\Gamma(\hat\theta), \T$)};
    \node [block, below=1cm of psim] (sim) {Simulate($\Gamma(\theta), \T$)};
    \node [block, below right=1.5cm of psim, minimum height=8em, yshift=1.74cm] (update) {UpdateRule($\cdot$)};
    \node [block, minimum width=2em, below=0.5cm of sim] (delay) {$\vartriangle$};
    \draw [->] (perturb) -- node {$\hat\theta(n)$} (psim);
    \draw [->] (noise) -- (perturb);
    \draw [->] (psim) -- node {$\hat X(n)$} (update.136);
    \draw [->] (sim) -- node {$X(n)$} (update.228);
    \draw [->] (update) -- node {$\theta(n + 1)$} +(3, 0);
%     \draw [->] (theta) -- +(1,0) |- (perturb);
%     \draw [->] (theta) -- +(1,0) |- node [very near end] {$\theta(n)$} (sim);
    \draw [->] (update) -- +(2, 0) |- (delay);
    \draw [->] (delay) -| node [near end] {$\theta(n)$} ($(perturb) - (2, 0)$) -- (perturb);
    \draw [->] (delay) -| ($(sim) - (psim) + (perturb) - (2, 0)$) -- node [very near end]  {$\theta(n)$} (sim);
\end{tikzpicture}}
\caption{Overall flow of the algorithm \ref{algorithm:sasoc-g-complete-algorithm}.}
\label{fig:sasoc-g-algorithm-flow}
\end{minipage}
\end{figure}

\begin{minipage}[c][\textheight]{\textwidth}
\begin{nonfloatalgorithm}{Skeleton of SASOC algorithms}
\label{algorithm:sasoc-g-complete-algorithm}
\begin{algorithmic}
\INPUT \ \\
\begin{itemize}[$\bullet$]
\item $R$, a large positive integer;
\item $\theta_0$, initial parameter vector; $p(\cdot)$; $\Delta$; $K \ge 1$
\item UpdateRule(), the algorithm-specific update rule for the worker parameter $\theta$ and Lagrange multiplier $\lambda$.
\item Simulate($\theta, \T$) $\rightarrow X$, the simulator of the SS
\end{itemize}
\OUTPUT $\theta^* \stackrel{\triangle}{=} \Gamma(\theta(R))$.
\hrule\vspace{1ex}
\STATE $\theta \leftarrow \theta_0$, $n \leftarrow 1$
\LOOP
%\FOR{$i = 1$ to $2$}
\STATE Observe $\hat J^{\theta} \leftarrow$ Simulate($\Gamma(\theta(n)),\T$).
\STATE $\hat{X} \leftarrow$ Simulate($\Gamma(\theta(n) + p(n)),\T$).
\STATE UpdateRule().
\STATE $n \leftarrow n + 1$
\IF{$n = R$}
\STATE Terminate and output $\Gamma(\theta(R))$.
\ENDIF
\ENDLOOP
\end{algorithmic}
\end{nonfloatalgorithm}
\end{minipage}

The overall flow of all SASOC algorithms can be diagrammatically represented as in Figure
\ref{fig:sasoc-g-algorithm-flow}. Each iteration of the algorithm involves two simulations (each for a period $\T$) - one with $\Gamma(\theta(n))$, i.e., the current estimate of the parameter projected using the generalized projection operator so that it takes values in the discrete set $\D$ and the other with the (projected) perturbed parameter, $\Gamma(\theta(n) + p(n))$, where the perturbation $p(n)$ is algorithm-specific. 
For instance, in the case of SASOC-G, $p(n) = \delta \Delta(n)$ and for SASOC-H/W, 
$p(n) = \delta_1 \Delta(n) + \delta_2 \hat\Delta(n)$, respectively. The rationale behind the choice of 
$p(n)$ will be subsequently clarified when the individual SASOC algorithms are presented.  In every stage of SASOC algorithms, the two simulations are carried out as shown in Figure \ref{fig:sasoc-g-algorithm-flow}. Using the state values of the two simulations, $X(n)$ and $\hat X(n)$, the worker parameter $\theta$ is updated in an algorithm-specific 
manner. Algorithm \ref{algorithm:sasoc-g-complete-algorithm} gives the structure of all three of our incremental update SASOC algorithms.

\subsection{SASOC-G Algorithm}
\label{sec:sasoc-g}
SASOC-G is a three time-scale stochastic approximation algorithm that does primal descent using a two-measurement SPSA while performing dual ascent on the Lagrange multipliers. 

\subsubsection{SPSA based gradient estimate}
Here, the gradient of the Lagrangian w.r.t. $\theta$ is obtained according to
\begin{align}
\label{eqn:gradestimate}
 \nabla_\theta L(\theta,\lambda) = \lim_{\delta\downarrow 0} E\left[ \left(
\frac{L(\theta +\delta\Delta,\lambda) -
L(\theta,\lambda)}{\delta}\right) \Delta^{-1}\right], 
\end{align}
where $\Delta$ is a vector (of the same dimension as $\theta$)
of perturbation random variables that are independent, zero-mean, $\pm$-valued and have the symmetric Bernoulli distribution. More general distributions on these random variables can be chosen as described in \citep{spall92multivariate,spall2000adaptive}. In \eqref{eqn:gradestimate}, $\Delta^{-1}$ represents element-wise inverse of the $\Delta$ vector. This is a one-sided estimate whose convergence is shown in \citep[Lemma 1]{chen1999kiefer}.
% For a sufficiently large $K > 1$ and small $\delta > 0$, the estimate of the gradient can be approximated as follows:
% % \vspace{-4ex}
% \begin{align*}
% \nabla_\theta L(\theta,\lambda) \approx \dfrac{1}{K} \sum_{i = 1}^K \left[ \left(
% \frac{L(\theta +\delta\Delta(i),\lambda) -
% L(\theta,\lambda)}{\delta}\right) \Delta(i)^{-1}\right]. 
% \end{align*}
% 
% % \vspace{-6ex}

\subsubsection{Update rule of SASOC-G}
From the form of the gradient estimator, it is clear that the Lagrangian function would be needed to compute the gradient estimate. However, for our problem, obtaining a closed-form expression for the Lagrangian itself is an intractable task.
We overcome this by running two simulations with parameters $\Gamma(\theta(n))$ and $\Gamma(\theta(n) + p(n))$. Here, $p(n) = \delta \Delta(n)$, a choice motivated by the form of the gradient estimate in \eqref{eqn:gradestimate}. Using the output of the two simulations, we estimate the quantities $L(\theta + \delta \Delta, \lambda)$ and $L(\theta, \lambda)$, respectively, on the faster timescale. These estimates are in turn used to tune the worker parameter $\theta$ in the negative gradient descent direction. For $\lambda_{i, j}$ and $\lambda_f$, values of $g_{i,j}(\cdot)$ and $h(\cdot)$ respectively provide a stochastic ascent direction, proof of which will be given later in Theorem \ref{theorem:sasoc-g-lambda}. Since maximization of the Lagrangian w.r.t. $\lambda_{i, j}$ and $\lambda_f$ represents the outer-most step, these parameters are updated on the slowest time-scale. The overall update rule for this scheme, SASOC-G, is as follows: For all $n \ge 0$,
\begin{equation}
\label{eqn:spsa-update-rule}
\left .\scalebox{0.96}{$\begin{array}{l}
\theta_{i}(n+1)  =  \bar\Gamma_i \left( \theta_{i}(n) + b(n)
\left(\frac{\bar{L}(nK) - \bar{L}'(nK)}{\delta \triangle_{i}(n)}
\right) \right),\forall i = 1, 2, \dots, N,\\[2ex]
\text{where for $m=0,1,\ldots,K-1$,}\\[2ex]
\bar{L}(nK+m+1) = \bar{L}(nK+m) + \\
d(n)(c(X_{nK+m}) + \sum\limits_{i=1}^{|C|}\sum\limits_{j=1}^{|P|} \lambda_{i,j}(nK) g_{i,j}(X_{nK+m}) + \lambda_f h(X_{nK+m}) - \bar{L}(nK+m) ), \\[4ex]
\bar{L}'(nK+m+1) = \bar{L}'(nK+m) + \\
d(n)(c(\hat{X}_{nK+m}) + \sum\limits_{i=1}^{|C|}\sum\limits_{j=1}^{|P|} \lambda_{i,j}(nK) g_{i,j}(\hat{X}_{nK+m}) + \lambda_f h(\hat{X}_{nK+m}) - \bar{L}'(nK+m) ),  \\[4ex]
\lambda_{i,j}(n+1) = \left( \lambda_{i,j}(n) + a(n)g_{i,j}(X_n) \right)^{+}, \forall i = 1, 2, \dots, |C|, j = 1, 2, \dots, |P|, \\[2ex]
\lambda_f(n+1) = \left( \lambda_f(n) + a(n)h(X_n) \right)^{+}.
\end{array}$} \right \}
\end{equation}
In the above,
\begin{itemize}[$\bullet$]
\item $K \ge 1$ is a fixed parameter which controls the rate of update of $\theta$ in relation to that of $\bar{L}$ and $\bar{L}'$. This parameter allows for accumulation of updates to $\bar{L}$ and $\bar{L}'$ for $K$ iterations in between two successive $\theta$ updates;
\item $X_m$ represents the state at iteration $m$ from the simulation run with nominal parameter $\Gamma(\theta_{[\frac{n}{K}]})$ while $\hat{X}_m$ represents the state at iteration $m$ from the simulation run with perturbed parameter $\Gamma(\theta_{[\frac{n}{K}]} + \delta \Delta_{[\frac{n}{K}]})$. Here $[\frac{n}{K}]$ denotes the integer portion of $\frac{n}{K}$. For simplicity, hereafter we use $\theta$ to denote $\theta_{[\frac{n}{K}]}$ and $\theta + \delta\Delta$ to denote $\theta_{[\frac{n}{K}]} + \delta \Delta_{[\frac{n}{K}]}$;
\item $\delta > 0$ is a fixed perturbation control parameter while $\Delta$ is a vector of perturbation random variables that are independent, zero-mean and have the symmetric Bernoulli distribution;
\item The operator $\bar\Gamma(\cdot)$ ensures that the updated value for $\theta$ stays within the convex hull $\bar\D$ and is defined in Section \ref{sec:gammaproj}; and
\item $\bar{L}$ and $\bar{L}'$ represent Lagrange estimates corresponding to $\theta$ and $\theta + \delta \Delta$ respectively.
% Thus, for each iteration, two simulations are carried out, one with the normal parameter $\theta$ and the other with the perturbed parameter $\theta + \delta \Delta$, the results of which are used to update $\bar{L}$ and $\bar{L}'$.
\end{itemize}

We achieve separation of time-scales between the recursions of $\theta_{i}, \bar{L}, \bar{L}'$ and $\lambda$ via the difference in the step-sizes $a(n), b(n)$ and $d(n)$ (see \textit{(A3)}).
The chosen step-sizes ensure that the recursions of Lagrange multipliers $\lambda_{i,j}$ proceed `slower' in comparison to those of the worker parameter $\theta$, while the updates of the average cost - $\bar{L}$ and $\bar{L}'$ proceed the fastest.
% The value of $K$ can be chosen arbitrarily as the analysis of convergence works for any $K \ge 1$. In practice, it is observed, see, for instance, \citep{bhatnagar2005adaptive,bhatnagar2007adaptive} that a value of $K$ between $50$ and $500$

\subsection{SASOC-H Algorithm}
\label{sec:sasoc-h}
This is  a second-order algorithm for adaptive labour staffing which uses SPSA based techniques to estimate both the gradient and the Hessian. As discussed before, the overall algorithm structure is represented by 
Figure \ref{fig:sasoc-g-algorithm-flow} with $p(n) = \theta(n) + \delta_1 \Delta(n) + \delta_2 \hat \Delta(n)$ in this case. Thus, each iteration of the algorithm involves two simulations (each for a period $\T$) - 
one with $\Gamma(\theta(n))$ and the other with $\Gamma(\theta(n) + \delta_1 \Delta(n) + 
\delta_2 \hat \Delta(n))$ as the respective parameters in the $n$th iteration cycle. As explained in 
the next section, the two perturbation sequences $\Delta$ and $\hat\Delta$ are used to 
estimate both the gradient and the Hessian of the Lagrangian w.r.t. $\theta$.

\subsubsection{SPSA based simultaneous estimates for the gradient and the Hessian}
Suppose the Lagrangian in (\ref{eqn:Lagrangian}) is twice differentiable w.r.t. $\theta$, then we can look at possible second order schemes for computing updates to $\theta$. 
If the Lagrangian (\ref{eqn:Lagrangian}) were a quadratic, then the exact solution for the
$\theta$ update to reach the minimum point would have been $-[\nabla^2_\theta L (\theta_0)]^{-1} \nabla_\theta L(\theta_0)$ with $\theta_0$ as the starting point, i.e.,  \[\theta^* = \theta_0 -[\nabla^2_\theta L (\theta_0)]^{-1} \nabla_\theta L(\theta_0), \] would be the optimal parameter. 
For a higher-degree Lagrangian, the above solution can be used with a step-size parameter iteratively till convergence to an optimal $\theta^*$. Let $\Delta$ and $\widehat\Delta$ be two independent 
vectors of perturbation random variables that are independent, zero-mean, $\pm 1$-valued and have the symmetric Bernoulli distribution. More general distributions for $\Delta$ and $\widehat\Delta$ may however be used, see \citep{spall92multivariate,
spall2000adaptive}. We use the following estimates for the gradient and Hessian, respectively, as described in \citep[Section 3.2.1]{shalabh2011constrained}:
\[ \nabla_\theta L(\theta,\lambda) = \lim_{\delta_1,\delta_2\downarrow 0} E\left[ \left(
\frac{L(\theta +\delta_1\Delta+\delta_2\widehat{\Delta},\lambda) -
L(\theta,\lambda)}{\delta_2}\right) \widehat{\Delta}^{-1}\right], \]
\[ \nabla_\theta^2 L(\theta,\lambda) = \lim_{\delta_1,\delta_2\downarrow 0}
E\left[ \Delta^{-1}
\left( \frac{L(\theta +\delta_1\Delta+\delta_2\widehat{\Delta},\lambda) -
L(\theta,\lambda)}{\delta_1\delta_2}\right) \left(\widehat{\Delta}^{-1}\right)^T\right], \]
where $\Delta^{-1}$ and $\widehat\Delta^{-1}$ represent vectors of element-wise inverses 
of the $\Delta$ and $\widehat\Delta$ vectors respectively. Thus, the inner terms of the above two expectations can be used for estimating the Hessian and also updating $\theta$.

\subsubsection{Update rule of SASOC-H}
% Following the rest as in the previous algorithm, we propose the update rule of SASOC-H, using two perturbation sequences.
For $n \ge 0$, we have

\begin{align}
\label{eqn:hessian-update-rule}
\theta_{i}(n+1)  =  \bar\Gamma_i \left( \theta_{i}(n) + b(n)
\sum\limits_{j = 1}^{N} M_{i, j}(n) \left(\dfrac{\bar{L}(nK) - \bar{L}'(nK)}{\delta_2 \widehat\triangle_{j}(n)}
\right) \right),\\
H_{i, j}(n + 1) = H_{i, j}(n) + b(n) \left ( \dfrac{\bar{L}'(nK) - \bar{L}(nK)}{\delta_1 \triangle_{j}(n) \delta_2 \widehat\triangle_{i}(n)} - H_{i, j}(n) \right ),
\end{align}
for $i,j=1,\ldots,N$. Note that
\begin{itemize}[$\bullet$]
\item The update equations corresponding to $\bar L, \bar{L}'$, $\lambda_{i,j}, i=1,\ldots,|C|, j=1,\ldots,|P|$ and $\lambda_f$ are the same as in SASOC-G \eqref{eqn:spsa-update-rule}. However, note that the perturbed parameter in this case is $(\theta(n) + \delta_1\Delta(n) + \delta_2 \widehat\Delta(n))$. Thus, 
unlike SASOC-G, $\hat{X}_m$ represents the state at iteration $m$ from the simulation run with 
perturbed parameter $\Gamma(\theta(n) + \delta_1\Delta(n) + \delta_2 \widehat\Delta(n))$, 
while $X_m$ continues to have the same interpretation as with SASOC-G.
\item $\delta_1, \delta_2 > 0$ are fixed perturbation control parameters while
$\Delta$ and $\widehat\Delta$ are two independent vectors of perturbation random
variables that are independent, zero-mean, $\pm 1$-valued, and have the symmetric
Bernoulli distribution;
\item $H = [H_{i,j}]_{i = 1,j = 1}^{|A|\times |B|, |A|\times |B|}$ represents
the Hessian (second-derivative w.r.t. $\theta$) estimate of the Lagrangian.
$H(0)$ is a positive definite and symmetric matrix. We let $H(0) =  \omega I$,
with $\omega > 0$ and $I$ being the identity matrix; and
\item $M(n) = \Upsilon(H(n))^{-1}= [M(n)_{i,j}]_{i = 1,j = 1}^{|A|\times |B|, |A|\times |B|}$ represents the inverse of the Hessian estimate $H$ of the Lagrangian, where $\Upsilon(\cdot)$ is a projection operator ensuring that the Hessian estimates remain symmetric and positive definite. The $\Upsilon$ operation is assumed to satisfy assumption \textit{(A4)}.
\end{itemize}

\vspace{4pt}
\noindent {\bf Assumption (A4)}

\vspace{4pt}
\noindent The projection operator $\Upsilon(\cdot)$ projects a square matrix to a symmetric positive definite matrix. If $\{A_n\}$ and $\{B_n\}$ are sequences of matrices in ${\cal R}^{N\times N}$
such that ${\displaystyle \lim_{n\rightarrow \infty} \parallel A_n-B_n \parallel}$ $= 0$,
then ${\displaystyle \lim_{n\rightarrow \infty} \parallel \Upsilon(A_n)- \Upsilon(B_n) \parallel}$
$= 0$
as well. Further, for any sequence $\{C_n\}$ of matrices in ${\cal R}^{N\times N}$,
if ${\displaystyle \sup_n \parallel C_n\parallel}$
$<\infty$,
then $\sup_n \parallel \Upsilon(C_n)\parallel < \infty$ and
$\sup_n \parallel \{\Upsilon(C_n)\}^{-1} \parallel <\infty,$ as well.

\vspace{4pt}

% We show in a later section with simulation results that second-order approaches exhibit better convergence behaviour.

\subsection{Efficient implementation of SASOC-H}
\label{sec:sasoc-w}

The SASOC-H algorithm is more robust than SASOC-G. However, it requires computation of the inverse of the Hessian $H$ at each stage which is a computationally intensive operation. We propose an enhancement using Woodbury's identity to the previous algorithm that results in significant computational gains. In particular, an application of Woodbury's identity brings down the computational complexity from $O(n^3)$\footnote{The popular 
Gauss-Jordan procedure for matrix inverse of a matrix requires $O(n^3)$ computations.} 
to $O(n^2)$ where $n = |A| \times |B|$.

\subsubsection{Woodbury's Identity based Update for Hessian Inverse}
  Woodbury's identity states that \[(A + BCD)^{-1} = A^{-1} - A^{-1} B \left ( C^{-1} + D A^{-1} B \right )^{-1} D A^{-1}\] where $A$ and $C$ are invertible square matrices and $B$ and $D$ are rectangular matrices of appropriate sizes. The Hessian update in (\ref{eqn:hessian-update-rule}) without projection can be rewritten as
\[H(n + 1) = (1 - b(n)) H(n) +  P(n) Z(nK) Q(n)\]
where \\
$P(n) = \dfrac{1}{\delta_1} \left [ \dfrac{1}{\Delta_1(n)},\dfrac{1}{\Delta_2(n)}, \ldots ,\dfrac{1}{\Delta_{|A| \times |B|}(n)} \right ]^T \kern-1ex, Q(n) = \dfrac{1}{\delta_2} \left [ \dfrac{1}{\widehat\Delta_1(n)}, \dfrac{1}{\widehat\Delta_2(n)}, \ldots ,\dfrac{1}{\widehat\Delta_{|A| \times |B|}(n)} \right ],$ and $Z(nK) = b(n) \left (\bar{L}'(nK) - \bar{L}(nK) \right )$.

Now, applying the Woodbury's identity to $H(n + 1)^{-1} = M(n + 1)$ gives us the following
update:
\[M(n + 1) = \left (\dfrac{M(n)}{1 - b(n)} \left [ I - \dfrac{b(n) \left ( \bar{L}'(nK) - \bar{L}(nK) \right ) P(n) Q(n) M(n)}{1 - b(n) + b(n) \left ( \bar{L}'(nK) - \bar{L}(nK) \right ) Q(n) M(n) P(n) } \right ] \right ),\] which is a recursive update rule for directly updating the matrix $M(n)$, which is the inverse of $H(n), n \ge 0$.

The modified update scheme of SASOC-H after incorporating the Woodbury's
identity for estimating the inverse of the Hessian, is as follows: For $n \ge
0$,
\begin{align}
\label{eqn:wudbury-update-rule}
\theta_{i}(n+1)  = & \bar\Gamma_i \left( \theta_{i}(n) + b(n)
\sum\limits_{j = 1}^{N} M_{i, j}(n) \left(\dfrac{\bar{L}(nK) - \bar{L}'(nK)}{\delta_1 \widehat\triangle_{j}(n)}
\right) \right),\\\nonumber
M(n + 1) = &\Upsilon \left (\dfrac{M(n)}{1 - b(n)} \left [ I - \dfrac{b(n) \left ( \bar{L}'(nK) - \bar{L}(nK) \right ) P(n) Q(n) M(n)}{1 - b(n) + b(n) \left ( \bar{L}'(nK) - \bar{L}(nK) \right ) Q(n) M(n) P(n) } \right ] \right ).
\end{align}
In the above, $M(0)$ is initialized to $\omega I$, $I$ being an identity matrix
and $\omega > 0$. 
The rest of the update rule corresponding to $\bar L, \bar{L}'$, $\lambda_{i,j},
i=1,\ldots,|C|, j=1,\ldots,|P|$ and $\lambda_f$ are the same as before (see
\eqref{eqn:spsa-update-rule}--\eqref{eqn:hessian-update-rule}).
% The SASOC-W algorithm is computationally better than SASOC-H though it retains its robustness.

% \begin{remark}
Our SASOC algorithms differ from the algorithms of \citep{shalabh2011constrained} in the following ways:
\begin{inparaenum}[(i)]
\item Unlike the algorithms of \citep{shalabh2011constrained} which are for a continuous-valued parameter, our SASOC algorithms are for a constrained discrete optimization setting and involve a generalized projection operator that renders the transition probabilities of the extended Markov process for any $\theta \in \bar\D$ smooth.
\item Since the SASOC-W algorithm does not involve explicit computation of the
Hessian inverse, it is computationally more efficient than the second-order
algorithms of \citep{shalabh2011constrained}.
\end{inparaenum}    
% \end{remark}

\section{Notes on convergence}
\label{sec:convergence}
Here we provide a sketch of the convergence of SASOC-G and SASOC-H
algorithms.\footnote{The detailed proofs of the various results are provided
in a supplementary file for review}.

\subsection*{Step 1: Extension of the transition dynamics $p_{i,j}(\theta)$}
The first step in the convergence analysis is common to both the
SASOC algorithms and involves the extension of the transition dynamics
$p_\theta(i,j)$ of the constrained parameterized Markov process to the
convex hull $\bar\D$.

Recall that the discrete parameter $\theta$ of the Markov process $\{X_n(\theta)\}$ takes values in the set $\D$ defined earlier. Using the members of $\D$, one can extend the transition dynamics $p_\theta(i,j)$ of the underlying Markov process to any $\theta$ in the convex hull $\bar\D$ as follows:
% one can write any $\theta \in \bar\D$ as
% \begin{equation}
%  \theta = \sum\limits_{k=1}^{N} \beta_k(\theta) D^k,
% \label{sasocalgos:eq:thetaD}
% \end{equation}
% where the weights $\beta_k(\theta)$ satisfy $0 \le \beta_k(\theta) \le 1, k=1,ldots,N$ and $\sum\limits_{k=1}^{N} \beta_k(\theta) = 1$.
% We now define the transition probability $p_\theta(i,j)$ for any $\theta \in \bar\D$ as
\begin{equation}
     p_\theta(i,j) = \sum\limits_{k=1}^{p} \beta_k(\theta) p_{D^k}(i,j), \quad \forall \theta \in \bar\D, i,j \in S,
\label{eq:pthetabar}
\end{equation}
where the weights $\beta_k(\theta)$ satisfy $0 \le \beta_k(\theta) \le 1, k=1,\ldots,p$ and $\sum\limits_{k=1}^{p} \beta_k(\theta) = 1$. For this choice of $\beta_k(\theta)$, $p_\theta(i,j), i,j \in S, \theta \in \bar\D$ can be seen to satisfy the properties of transition probabilities. We now explain the manner in which these weights are obtained.  It is worth noting here that the weights $\beta_k(\theta)$ must be continuously differentiable in order to ensure that the extended transition probabilities are continuously differentiable as well and our SASOC algorithms converge. Moreover, in the SASOC algorithms, we do not require an explicit computation of these weights while trying to solve the constrained optimization problem \eqref{eqn:c-mp}. Consider the case when $\theta = \theta_1$ and suppose $\theta_1$ lies between $D^j$ and $D^{j+1}$ (both members of $\D$). By construction, $\beta_k(\theta_1)$ will correspond to the probability with which projection is done on $[D^j,D^{j+1}]$ and is obtained 
using the $\Gamma$-projection operator as follows:  Let us consider an interval of length $2 \zeta$ around the midpoint of $[D^j, D^{j+1}]$ and denote it as $[\tilde D_1, \tilde D_2]$, where $\tilde D_1 = \frac{\D^j+\D^{j+1}}{2} - \zeta$ and $\tilde D_2 = \frac{\D^j+\D^{j+1}}{2} + \zeta$.
Then, the weights $\beta_k(\theta_1)$ are set in the following manner:
$\beta_k(\theta_1) = 0, \forall k \notin \{j,j+1\}$ and $\beta_j(\theta_1), \beta_{j+1}(\theta_1)$ is given by:

\begin{equation}
 (\beta_j(\theta_1),\beta_{j+1}(\theta_1)) =
  \begin{cases}
   (1,0) &  \text{if } \theta_1 \in\left[D^j,\tilde D_1 \right] \\
   (f(\frac{\tilde D_2 - \theta_1}{2\zeta}), 1- f(\frac{\tilde D_2 - \theta_1}{2\zeta})) & \text{if } \theta_1 \in \left[\tilde D_1, \tilde D_2\right]  \\
   (0,1) & \text{if } \theta_1 \in \left[\tilde D_2,D^{j+1}\right]
  \end{cases}
\end{equation}
In the above, $f$ to be a continuously differentiable function defined on $[0,1]$ such that $f(0)=0$ and $f(1)=1$ and the $\Gamma$-projection is derived from such an $f$. The above can be similarly extended when the parameter $\theta$ has $N$ components. It can thus be seen that $\beta_k(\theta), k=1,\ldots,p$ are continuously differentiable functions of $\theta$. Thus, from  \eqref{eq:pthetabar} and the fact that $\beta_k(\theta)$ are continuously differentiable, it can be seen that the extended transition dynamics $p_\theta(i,j), \forall \theta \in \bar\D, i,j \in S$ are continuously differentiable.
We now claim the following:
\begin{lemma}
\label{lemma:sasocequivalence}
Under the extended dynamics $p_\theta(i,j),i,j \in S$ of the Markov process
$\{X_n(\theta)\}$ defined over all $\theta \in \bar \D$, we have
\begin{enumerate}[(i)]
    \item SASOC-G algorithm is analogous to its continuous counterpart where $\Gamma(\theta)$ and $\Gamma(\theta + \delta \triangle)$ are replaced by $\bar\Gamma(\theta)$ and $\bar\Gamma(\theta + \delta \triangle)$ respectively.
    \item SASOC-H algorithm is analogous to its continuous counterpart
where $\Gamma(\theta)$ and $\Gamma(\theta + \delta_1 \triangle + \delta_2
\hat\triangle)$ are replaced by $\bar\Gamma(\theta)$ and $\bar\Gamma(\theta
+\delta_1 \triangle + \delta_2 \hat\triangle)$ respectively.
\end{enumerate}
\end{lemma}

\subsection*{Step 2: Analysis of fastest timescale recursion}
The fastest time-scale in SASOC-G is $\{d(n)\}$ which is used
to update the Lagrangian estimates $\bar{L}$ and $\bar{L}'$ corresponding to
simulations with $\theta$ and $\theta + \delta \Delta$ respectively. First, we
show that these estimates indeed converge to the Lagrangian values $L(\theta,
\lambda)$ and $L(\theta + \delta \Delta, \lambda)$ defined in
(\ref{eqn:gradestimate}). By the choice of
step-sizes satisfying (A3), we have a time-scale separation between the updates
to the Lagrangian estimates $\bar{L}_n$ and the parameters - $\theta$ and
$\lambda$. Hence, for the purpose of analysis of these
Lagrangian estimates, $\theta$ and $\lambda$ can be
assumed to be time invariant quantities. We now have the following result:

\begin{lemma}
\label{lemma:Lagrangian}
\begin{enumerate}[(i)]
    \item For SASOC-G algorithm, $\|\bar{L}(n) - L(\theta(n), \lambda(n)) \|
\rightarrow 0$ w.p. 1, as $n \rightarrow \infty$.
    \item For SASOC-H algorithm, 
$
\|\bar{L}(n) - L(\theta(n), \lambda(n)) \|, \|\bar{L}'(n) -
L(\theta(n) + \delta_1 \Delta(n) + \delta_2 \widehat\Delta(n), \lambda(n)) \|
\rightarrow 0 \textrm{ as } n \rightarrow \infty.
$
\end{enumerate}
\end{lemma}

\subsection*{Step 3: Analysis of the $\theta$-recursion}
We show that the evolution of $\theta$ in SASOC-G descends in the
Lagrangian value and converges to a limiting set that depends on $\lambda$. For
this purpose, we first show that the resulting martingale from the $\theta$
update recursion in (\ref{eqn:spsa-update-rule}) is convergent and then use
$V^{\lambda}(\cdot) = L (\theta,\lambda)$ as an associated Lyapunov function for
the following ODE 
\begin{equation}
\label{eqn:sasoc-g:theta-ode}
\dot{\theta}(t) = \check{\Gamma}\left ( -\nabla_\theta L(\theta(t), \lambda)
\right
),
\end{equation}
where $\check{\Gamma}$ is defined as follows: For any bounded continuous
function
$\epsilon(\cdot)$,
\begin{equation}
\label{eqn:Pi-bar-operator}
\check{\Gamma}(\epsilon(\theta(t))) = \lim\limits_{\eta \downarrow 0}
\dfrac{\Gamma(\theta(t) + \eta \epsilon(\theta(t))) - \theta(t)}{\eta}.
\end{equation}
The projection operator $\check{\Gamma}(\cdot)$ ensures that the evolution of
$\theta$ via the ODE (\ref{eqn:sasoc-g:theta-ode})
stays within the bounded set $\bar\D$. Again for the analysis of the
$\theta$-update, the value of $\lambda$ which is updated on the slowest
time-scale is assumed constant.

\begin{theorem}
Under (A1), (A2') and (A3), with $\lambda(n)\equiv\lambda,\forall n$, in the limit as $\delta \rightarrow 0$,
$\theta(R) \rightarrow \theta^* \in K^{\lambda}$ almost surely as $R \rightarrow
\infty$, where $K^{\lambda} = \{
\theta \in S: \check{\Gamma}\left ( -\nabla L(\theta(t), \lambda) \right ) = 0
\}$.
\end{theorem}

Similarly, we show that the parameter updates
$\theta(n)$ of SASOC-H converge to a limit point of the ODE
\begin{equation}
\label{eqn:sasoc-h:theta-ode}
\dot{\theta}(t) = \check{\Gamma}\left ( - \Upsilon(\nabla^2_\theta L(\theta(t),
\lambda))^{-1} \nabla_\theta L(\theta(t), \lambda) \right ).
\end{equation}

\begin{theorem}
Under (A1), (A2), (A3) and (A4), with $\lambda(n)\equiv\lambda,\forall n$,
in the limit as $\delta_1, \delta_2 \rightarrow
0$,
$\theta(R) \rightarrow \theta^* \in \bar{K}^{\lambda}$ almost surely as $R
\rightarrow \infty$, where \[\bar{K}^\lambda = \left \{ \theta \in S:
\dfrac{d L (\theta(t),
\lambda)}{dt} = - \nabla_{\theta} L (\theta(t), \lambda)^T
\Upsilon(\nabla^2_\theta L(\theta(t), \lambda))^{-1} \nabla_\theta L(\theta(t),
\lambda) = 0 \right \}.\]
\end{theorem}

Note that $K^\lambda$ and $\bar{K}^\lambda$ can differ in spurious fixed points
on the boundary of $\bar{\D}$.

\subsection*{Step 4: : Analysis of the $\lambda$-recursion}
For $\{\lambda(n)\}$ updates on the slowest time-scale $\{a(n)\}$, we can assume
that $\theta$ has converged to $\theta^* \in K^\lambda$. We show that
$\lambda_{i, j}$s and $\lambda_f$ converge respectively
to the limit points of the ODEs \[\begin{array}{l}
\dot{\lambda}_{i,j}(t) = \check\Pi \left ( G_{i, j}(\theta^*) \right ), \forall
i
=
1, 2, \dots, |C|, j = 1, 2, \dots, |P|,\\[1ex]
\dot{\lambda}_f(t) = \check\Pi \left ( H(\theta^*) \right ),
\end{array}\]
where $\theta^*$ is the converged parameter value of SASOC-G/H corresponding to
Lagrange parameter $\lambda(t) \stackrel{\triangle}{=} (\lambda_{i,j}(t),
\lambda_f(t), i=1,\ldots,|C|, j=1,\ldots,|P|)^T$, and for any bounded continuous
functions $\bar{\epsilon}(\cdot)$, \[\check\Pi(\bar{\epsilon}(\lambda(t))) =
\lim\limits_{\eta \downarrow 0} \dfrac{(\lambda(t) + \eta
\bar{\epsilon}(\lambda(t)))^+ - \lambda(t)}{\eta}.\] 
Here again, the projection
operator $\check\Pi$ ensures that the evolution of each component
of $\lambda$ stays non-negative.
From the definition of the Lagrangian given in (\ref{eqn:Lagrangian}),
the gradient of the Lagrangian w.r.t. $\lambda_{i, j}$ can be seen to  be
$G_{i,j}(\theta^*)$ and that w.r.t. $\lambda_f$ is $H(\theta^*)$. Thus, the
above ODEs suggest that in SASOC-G/H $\lambda_{i, j}$s' and $\lambda_f$ are
ascending in the Lagrangian value and converge to a local maximum point.
We now have the following result:
 
\begin{theorem}
\label{theorem:sasoc-g-lambda}
Let $F^{\theta^*} =
\left \{ \lambda \ge 0 : \check\Pi \left ( G_{i, j}(\theta^*) \right ) = 0,
\forall
i = 1, 2, \dots, |C|, j = 1, 2, \dots, |P|; \check\Pi \left ( H(\theta^*) \right
)
= 0 \right \}.$ Then, $\lambda(R) \rightarrow \lambda^*$ for some $\lambda^*
\in F^{\theta^*}$ w.p. 1 as $R \rightarrow \infty$.
\end{theorem}

\subsection*{Step 5: Convergence to a locally saddle point}

Finally, we argue that the algorithm indeed converges to a (local) saddle
point of the Lagrangian.
Suppose $H_1$ denote a local neighborhood in which $\theta^*$ is
a minimum. Then, through an application of the envelope theorem of mathematical economics \citep[pp.
964-966]{mas1995microeconomic}, applied in the `Caratheodory sense' \citep[Lemma 4.3, pp.211]{borkar2005actor},
it can be seen that
\[\lambda^* \in \arg\min_{\lambda\in H_2} \min_{\theta\in H_1} L(\theta,\lambda),\]
where $H_2$ is some local neighborhood that contains $\lambda^*$.
The SASOC algorithms thus converge to a locally saddle point. As mentioned
at the beginning of this section, the detailed proofs of the above results
are available in an attached supplementary file.

%%%%%%%%%%%%%%%%%%%%%%%%%%%%%%%%%%%%%%%%%%%%%%%%%%%%%%%%%%%%
\section{Simulation Experiments}
\label{sec:simulation}
% \subsection{Implementation}

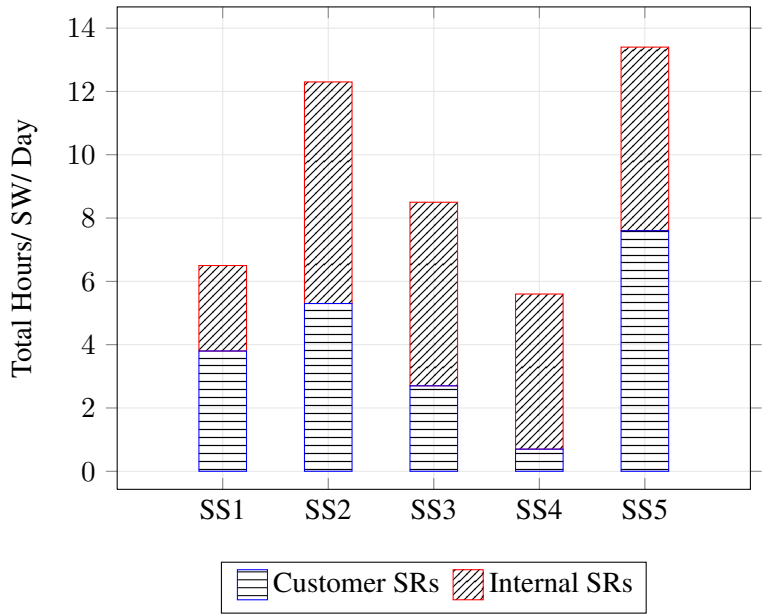
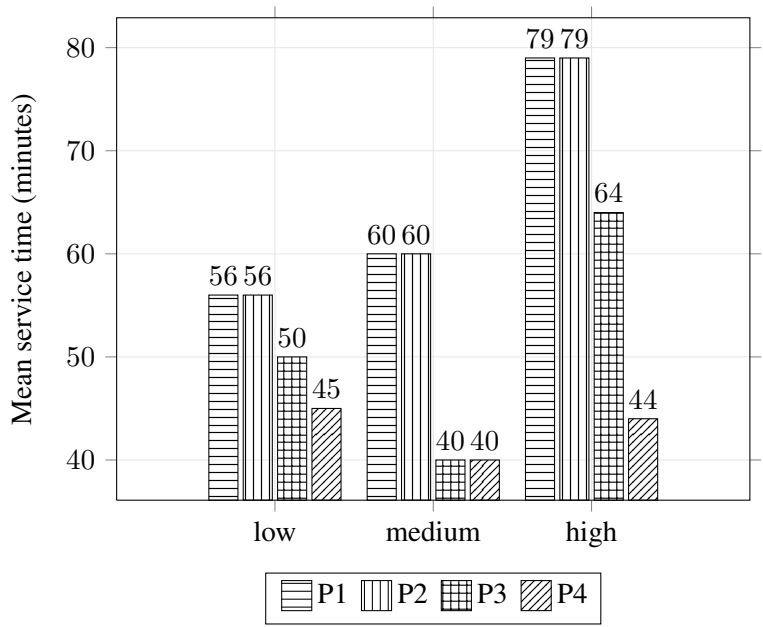
\begin{figure}
\centering
\begin{tabular}{l}
    \subfigure[Total work volume statistics for each SS]
    {
	
	\tabl{c}{\scalebox{1.0}{\begin{tikzpicture}
	\begin{axis}[
	ybar stacked,
	legend style={at={(0.5,-0.15)},anchor=north,legend columns=-1},
	legend image code/.code={\path[fill=white,white] (-2mm,-2mm) rectangle (-3mm,2mm); \path[fill=white,white] (-2mm,-2mm) rectangle (2mm,-3mm); \draw (-2mm,-2mm) rectangle (2mm,2mm);},
	ylabel={Total Hours/ SW/ Day},
	symbolic x coords={1, SS1, SS2, SS3, SS4, SS5, 2},
	xmin={1},
	xmax={2},
	xtick=data,
	ytick align=outside,
	bar width=18pt,
	% nodes near coords,
	%nodes near coords align={vertical},
	grid,
	grid style={gray!20},
	width=10cm,
	height=8cm,
	]
	\addplot+[ybar,pattern=horizontal lines] coordinates {(SS1,3.8) (SS2,5.3) (SS3,2.7) (SS4,0.7) (SS5,7.6)}; % Customer work
	\addplot+[ybar,pattern=north east lines] coordinates {(SS1,2.7) (SS2,7) (SS3,5.8) (SS4,4.9) (SS5,5.8)}; % Internal work
	\legend{Customer SRs, Internal SRs}
	\end{axis}
	\end{tikzpicture}}\\[1ex]}
            \label{fig_average_workload}
    } \\
    \subfigure[Estimated mean service times for a SS]
    {
	\tabl{c}{\scalebox{1.0}{\begin{tikzpicture}
	\begin{axis}[
	ybar={2pt},
	% ybar stacked,
	legend style={at={(0.5,-0.15)},anchor=north,legend columns=-1},
	legend image code/.code={\path[fill=white,white] (-2mm,-2mm) rectangle (-3mm,2mm); \path[fill=white,white] (-2mm,-2mm) rectangle (2mm,-3mm); \draw (-2mm,-2mm) rectangle (2mm,2mm);},
	ylabel={Mean service time (minutes)},
	symbolic x coords={1, low, medium, high, 2},
	xmin={1},
	xmax={2},
	xtick=data,
	ytick align=outside,
	bar width=11pt,
	nodes near coords,
	%nodes near coords align={vertical},
	grid,
	grid style={gray!20},
	width=10cm,
	height=8cm,
	]
	\addplot[pattern=horizontal lines] coordinates {(low,56) (medium,60) (high,79)}; %P1
	\addplot[pattern=vertical lines]   coordinates {(low,56) (medium,60) (high,79)}; %P2
	\addplot[pattern=grid]             coordinates {(low,50) (medium,40) (high,64)}; %P3
	\addplot[pattern=north east lines] coordinates {(low,45) (medium,40) (high,44)}; %P4
	\legend{P1, P2, P3, P4}
	\end{axis}
	\end{tikzpicture}}\\[1ex]}
	\label{fig_servicetime}
    } 
\end{tabular}
\caption{Characteristics of the service systems used for simulation}
\end{figure}
 \begin{figure}
    \centering
    \begin{tabular}{c}
    \subfigure[SS1 and SS2 work arrival pattern]
    {
                \begin{tabular}{l}
            \tabl{c}{\includegraphics[width=4.5in]{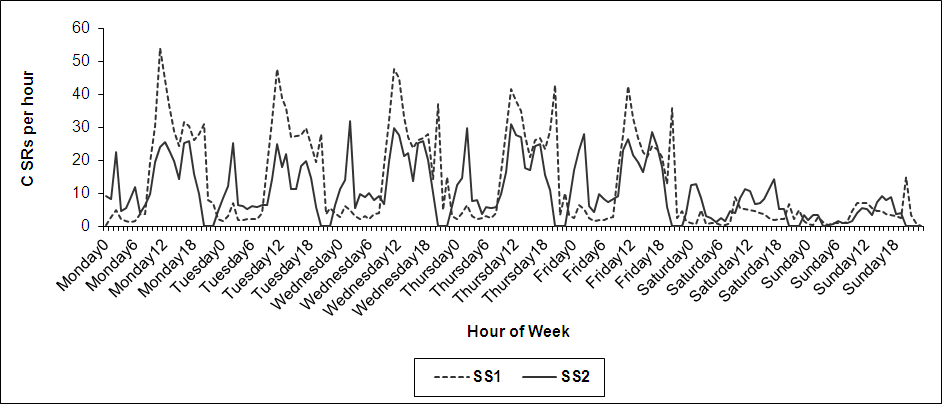}\\[2ex]}
            \end{tabular}
            \label{fig_arrivals_ss123}
    }\\ 
    \subfigure[SS3, SS4 and SS5 work arrival pattern]
    {
                \begin{tabular}{l}
            \tabl{c}{\includegraphics[width=4.5in]{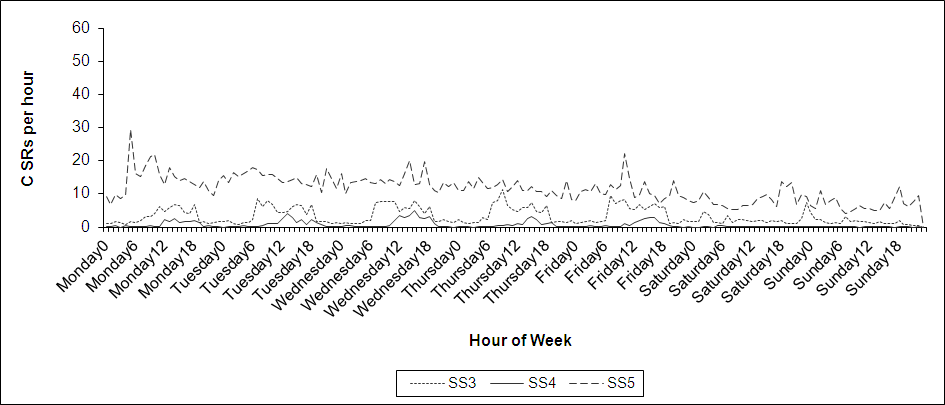}\\[2ex]}
            \end{tabular}
            \label{fig_arrivals_ss45}
    }
    \end{tabular}
% \vspace{-4ex}
    \caption{Work arrival patterns over a week for each SS}
    \label{fig_arrivals}
\end{figure}

% \vspace{-4ex}
We use the simulation framework developed in
\citep{banerjee2011simulation} for implementing all our algorithms. 
A number of dispatching policies have been developed in
\citep{banerjee2011simulation}. In particular, we study the PRIO-PULL
and EDF policies for performance comparisons of the various
algorithms. In addition to the three SASOC algorithms, we implemented an
algorithm that uses the state-of-the-art optimization tool-kit OptQuest, for the sake of comparison. 
% We implemented the following labor staffing algorithms:
% \begin{itemize}[$\bullet$]
%     \item \textbf{SASOC-G}: This is a first order method that estimates $\nabla L(\theta,\lambda)$ using SPSA and is  described in Section \ref{sec:sasoc-g}.
%     \item \textbf{SASOC-H}: This is a second order Newton method that involves an explicit inverse of the Hessian matrix and is described in Section \ref{sec:sasoc-h}.
%     \item \textbf{SASOC-W}: This is a second order Newton method, which unlike SASOC-H does away with the inversion of the Hessian matrix and leverages the Woodbury identity in order to estimate the inverse of the Hessian directly. This algorithm is described in Section \ref{sec:sasoc-w}.
%     \item \textbf{OptQuest}: This is an algorithm that uses the state-of-the-art optimization tool-kit OptQuest. In particular, we used the scatter search based variant of OptQuest for our experiments.\end{itemize}
OptQuest employs an
array of techniques including scatter and tabu search, genetic
algorithms, and other meta-heuristics for the purpose of optimization
and is quite well-known as a hybrid search tool for solving simulation optimization
problems (\citep{laguna1998optimization}). In particular, we have
used the scatter search variant of OptQuest for our experiments.
% OptQuest along with several
% other engines from Frontline Systems won the INFORMS impact
% award\footnote{http://www.solver.com/press201008.htm} in the year 2010.

We choose five real-life SS from two different countries providing
server support to IBM's customers. The five SS cover a
variety of characteristics such as high vs. low workload, small
vs. large number of customers to be supported, small vs. big staffing
levels, and stringent vs. lenient SLA constraints. Collectively, these five SS staff more than 200 SWs with $40\%$, $30\%$, and $30\%$ of them having low, medium, and high skill level, respectively. Also, these SS support more than $30$ customers each, who make more than $6500$ SRs every week with each customer having a distinct pattern of arrival depending on its business hours and seasonality of business domain.
 Figure \ref{fig_average_workload} shows the total work hours per SW per day
for each of the SS. The bottom part of the bars denotes customer SR work,
i.e., the SRs raised by the customers whereas the top part of the bars
denotes internal SR work, i.e, the SRs raised internally for overhead work
such as meetings, report generation, and HR activities.  This
segregation is important because the SLAs apply only to customer SRs. Internal SRs
do not have deadlines but they may contribute to queue growth.  Note
that while average work volumes are significant, they may not directly
correlate to SLA attainment. Figure \ref{fig_servicetime} shows the effort data,
i.e., the mean time taken to resolve an SR (a lognormal distributed random variable in our setting)
across priority and complexity classes.
 As shown in Figures \ref{fig_arrivals}, the
arrival rates for SS4 and SS5 show much higher peaks than SS1, SS2,
and SS3, respectively, although their average work volumes are comparable. The
variations are significant because during the peak periods, many SRs
may miss their SLA deadlines and influence the optimal staffing
result.

Some of the specific details of the service system
setting (see Section \ref{sec:formulation}) are as follows:
$\I$, the set of time intervals, contains one element for each hour of the
week. Hence, $|\I|$ = $168$. 
The set of priority levels, $P=\{P_1, P_2, P_3, P_4\}$, where, $P_1 > P_2 >
P_3 > P_4$. The set of skill levels $\B$ is $\{\textrm{High, Medium, Low}\}$,
where, High $>$ Medium $>$ Low. The simulation framework also involves the
swing and preemption policies and the reader is referred to
\citep{banerjee2011simulation} for a detailed description of this.

We implemented our
SASOC algorithms on the simulation framework from
\citep{banerjee2011simulation} for both the perturbed and the
unperturbed simulations (see $X$ and $\hat{X}$ computations in
Algorithm~\ref{algorithm:sasoc-g-complete-algorithm}).
For our SASOC algorithms, the
simulations were conducted for $1000$ iterations, with each iteration
having $20$ simulation replications - ten each with unperturbed
parameter $\theta$ and perturbed parameter $\hat\theta$,
respectively.  Each replication simulated the operations of the
respective SS for a $30$ day period.  Thus, we set $R=1000$ and $K=10$
for SASOC algorithms. On the other hand, for the OptQuest algorithm,
simulations were conducted for $5000$ iterations, with each iteration
of $100$ replications of the SS.

For all the SASOC algorithms, we set the weights
in the single-stage cost function $c(X_m)$, see
(\ref{eqn:singlestagecost}), as $r = s = 0.5$. We thus give equal
weightage to both the worker utilization and the SLA over-achievement
components. The indicator variable $q$ used in the constraint
(\ref{eqn:feasibility-constraint}) was set to $0$ (i.e., infeasible)
if the queues were found to grow by $1000\%$ over a two-week period
during simulation. We performed a sensitivity study for the paramter $\delta$
and found that the choice of $0.5$ gave the best results. For the second order
methods, the perturbation control parameters $\delta_1$ and $\delta_2$ were
both set to $0.5$. The function $f$ in the generalized projection operator was
set as $f(x) = x$, with the parameter $\zeta = 0.1$. Each of the experiments were run on a machine with dual core
Intel
$2.1$ GHz processor and $3$ GB RAM.

The
$\Upsilon$ operator implemented for SASOC-W can be described as
follows. Let $\hat{H}$ be the Hessian update which needs to be
projected. The following sequence of operations represent this
projection. \begin{inparaenum}[(i)] \item $\hat{H} \leftarrow
  \frac{(\hat{H} + \hat{H}^T)}{2}$; \item Perform eigen-decomposition
  on $\hat{H}$ to get all eigen-values and corresponding eigen-vectors; \item Project each eigen-value to $\left [ \epsilon,
    \frac{1}{\epsilon} \right ]$ where $1 > \epsilon > 0$. $\epsilon$
  is chosen to be a small number so as to allow for larger range of
  values, but not too small to avoid singularity. The upper limit in
  the projection range is to avoid singularity of the inverse of the
  Hessian estimate; and \item Reconstruct $\hat{H}$ using the
  projected eigen-values but with same eigen-vectors. \end{inparaenum}
The $\Upsilon$ operator in the case of SASOC-H
with diagonal Hessian is one that simply projects each diagonal entry to $\left [
  \epsilon, \frac{1}{\epsilon} \right ]$. It is easy to see that the
$\Upsilon$ operator satisfies assumption \textit{(A4)}. For a closely related
modification of the Hessian, the reader is referred to \citep{gill1981practical}.
In our experiments, we
set $\epsilon = 0.01$.

On each SS, we compare our SASOC algorithms with the OptQuest
algorithm using $W_{sum}$ and mean utilization as the performance metrics. Here $W_{sum}
\stackrel{\triangle}{=} \sum_{i=1}^{|A|}\sum_{j=1}^{|B|} \theta_{i,j}$ is
the sum of workers across shifts and skill levels. The mean utilization here refers to a weighted average of the utilization percentage achieved for each skill level, with the weights being the fraction of the workload corresponding to each skill level.

As evident in Figures \ref{fig_arrivals_ss123} and \ref{fig_arrivals_ss45}, the SS pools SS1, SS2 and SS3 are characterized by a flat SR arrival pattern, whereas SS4 and SS5 are characterized by a bursty SR arrival pattern. We present and analyze the results on these pools separately, starting with the flat arrival pools in the next section.

\subsection{Flat-Arrival SS pools}
% \vspace{-4ex}

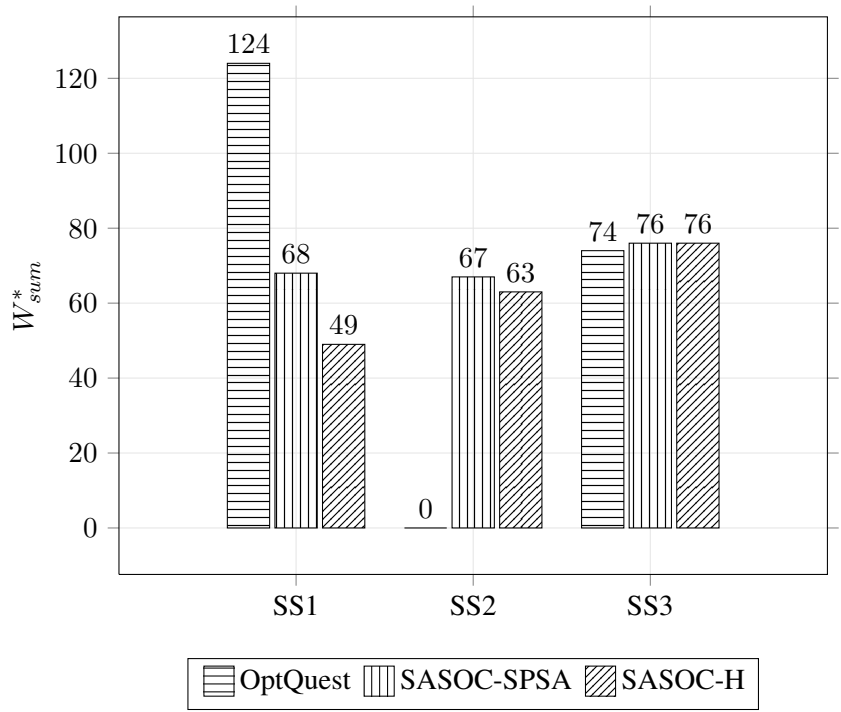
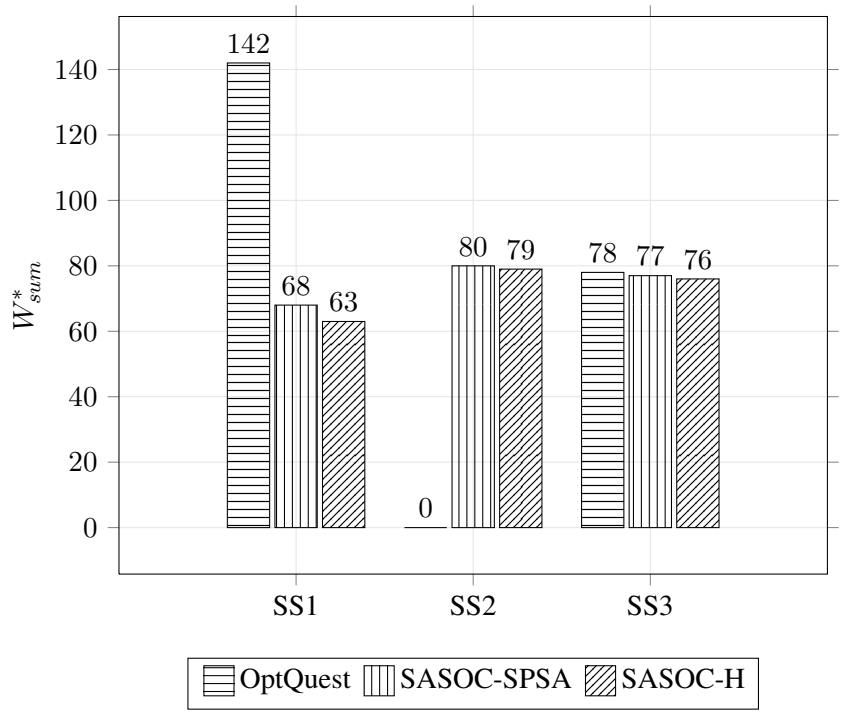
\begin{figure}
\begin{minipage}[c][\textheight]{\textwidth}
    \centering
    \begin{tabular}{c}
    \subfigure[$W^*_{sum}$ for PRIO-PULL]
    {
\tabl{c}{\scalebox{1.0}{\begin{tikzpicture}
\begin{axis}[
ybar={2pt},
legend style={at={(0.5,-0.15)},anchor=north,legend columns=-1},
legend image code/.code={\path[fill=white,white] (-2mm,-2mm) rectangle (-3mm,2mm); \path[fill=white,white] (-2mm,-2mm) rectangle (2mm,-3mm); \draw (-2mm,-2mm) rectangle (2mm,2mm);},
ylabel={$W_{sum}^*$},
symbolic x coords={1, SS1, SS2, SS3, 2},
xmin={1},
xmax={2},
xtick=data,
ytick align=outside,
bar width=16pt,
nodes near coords,
%nodes near coords align={vertical},
grid,
grid style={gray!20},
width=11cm,
height=9cm,
]
\addplot[pattern=horizontal lines] coordinates {(SS1,124) (SS2,0) (SS3,74)}; % OptQuest
\addplot[pattern=vertical lines]   coordinates {(SS1,68) (SS2,67) (SS3,76)}; % SASOC-SPSA
% \addplot[pattern=grid]             coordinates {(SS1,49) (SS2,79) (SS3,79)}; % SASOC-H
\addplot[pattern=north east lines] coordinates {(SS1,49) (SS2,63) (SS3,76)};
%SASOC-W
%
\legend{OptQuest, SASOC-SPSA, SASOC-H}
\end{axis}
\end{tikzpicture}}\\[1ex]}
\label{fig_wsum_priopull_ss123}
    } \\
    \subfigure[$W^*_{sum}$ for EDF]
    {
\tabl{c}{\scalebox{1.0}{\begin{tikzpicture}
\begin{axis}[
ybar={2pt},
legend style={at={(0.5,-0.15)},anchor=north,legend columns=-1},
legend image code/.code={\path[fill=white,white] (-2mm,-2mm) rectangle (-3mm,2mm); \path[fill=white,white] (-2mm,-2mm) rectangle (2mm,-3mm); \draw (-2mm,-2mm) rectangle (2mm,2mm);},
ylabel={$W_{sum}^*$},
symbolic x coords={1, SS1, SS2, SS3, 2},
xmin={1},
xmax={2},
xtick=data,
ytick align=outside,
bar width=16pt,
nodes near coords,
%nodes near coords align={vertical},
grid,
grid style={gray!20},
width=11cm,
height=9cm,
]
\addplot[pattern=horizontal lines] coordinates {(SS1,142) (SS2,0) (SS3,78)}; % OptQuest
\addplot[pattern=vertical lines]   coordinates {(SS1,68) (SS2,80) (SS3,77) }; % SASOC-SPSA
% \addplot[pattern=grid]             coordinates {(SS1,68) (SS2,82) (SS3,75) }; % SASOC-H
\addplot[pattern=north east lines] coordinates {(SS1,63) (SS2,79) (SS3,76) }; % SASOC-W

\legend{OptQuest, SASOC-SPSA, SASOC-H}
\end{axis}
\end{tikzpicture}}\\[1ex]}

                                \label{fig_wsum_edf_ss123}
    }    \end{tabular}
    \caption[Performance of OptQuest and SASOC algorithms]{Performance of OptQuest and SASOC algorithms on SS1, SS2 and SS3\footnote{Note: OptQuest is infeasible over SS2}}
    \label{fig_priopull_ss123}
\end{minipage}
\end{figure}
\begin{figure}
    \centering
    \begin{tabular}{c}
    \subfigure[$W^*_{sum}$ for PRIO-PULL]
    {
\tabl{c}{\scalebox{1.0}{\begin{tikzpicture}
\begin{axis}[
ybar={2pt},
legend style={at={(0.5,-0.15)},anchor=north,legend columns=-1},
legend image code/.code={\path[fill=white,white] (-2mm,-2mm) rectangle (-3mm,2mm); \path[fill=white,white] (-2mm,-2mm) rectangle (2mm,-3mm); \draw (-2mm,-2mm) rectangle (2mm,2mm);},
ylabel={$W_{sum}^*$},
symbolic x coords={1, SS4, SS5, 2},
xmin={1},
xmax={2},
xtick=data,
ytick align=outside,
bar width=16pt,
nodes near coords,
%nodes near coords align={vertical},
grid,
grid style={gray!20},
width=11cm,
height=9cm,
]
\addplot[pattern=horizontal lines] coordinates {(SS4,35) (SS5,39) }; % OptQuest
\addplot[pattern=vertical lines]   coordinates {(SS4,44) (SS5,46) }; % SASOC-SPSA
% \addplot[pattern=grid]             coordinates {(SS4,57) (SS5,57) }; % SASOC-H
\addplot[pattern=north east lines] coordinates {(SS4,53) (SS5,55) }; % SASOC-W
\legend{OptQuest, SASOC-SPSA, SASOC-H}
\end{axis}
\end{tikzpicture}}\\[1ex]}

\label{fig_wsum_priopull_ss45}
    } \\
    \subfigure[$W^*_{sum}$ for EDF]
    {
\tabl{c}{\scalebox{1.0}{\begin{tikzpicture}
\begin{axis}[
ybar={2pt},
legend style={at={(0.5,-0.15)},anchor=north,legend columns=-1},
legend image code/.code={\path[fill=white,white] (-2mm,-2mm) rectangle (-3mm,2mm); \path[fill=white,white] (-2mm,-2mm) rectangle (2mm,-3mm); \draw (-2mm,-2mm) rectangle (2mm,2mm);},
ylabel={$W_{sum}^*$},
symbolic x coords={1,  SS4, SS5, 2},
xmin={1},
xmax={2},
xtick=data,
ytick align=outside,
bar width=16pt,
nodes near coords,
%nodes near coords align={vertical},
grid,
grid style={gray!20},
width=11cm,
height=9cm,
]
\addplot[pattern=horizontal lines] coordinates {(SS4,121) (SS5,76)}; % OptQuest
\addplot[pattern=vertical lines]   coordinates {(SS4,45) (SS5,45) }; % SASOC-SPSA
% \addplot[pattern=grid]             coordinates {(SS4,72) (SS5,64) }; % SASOC-H
\addplot[pattern=north east lines] coordinates {(SS4,56) (SS5,57) }; % SASOC-W

\legend{OptQuest, SASOC-SPSA, SASOC-H}
\end{axis}
\end{tikzpicture}}\\[1ex]}

                                \label{fig_wsum_edf_ss45}
    }    \end{tabular}
    \caption{Performance of OptQuest and SASOC for two different dispatching
policies on SS4 and SS5}
    \label{fig_priopull_ss45}
% \vspace{-4ex}
\end{figure}
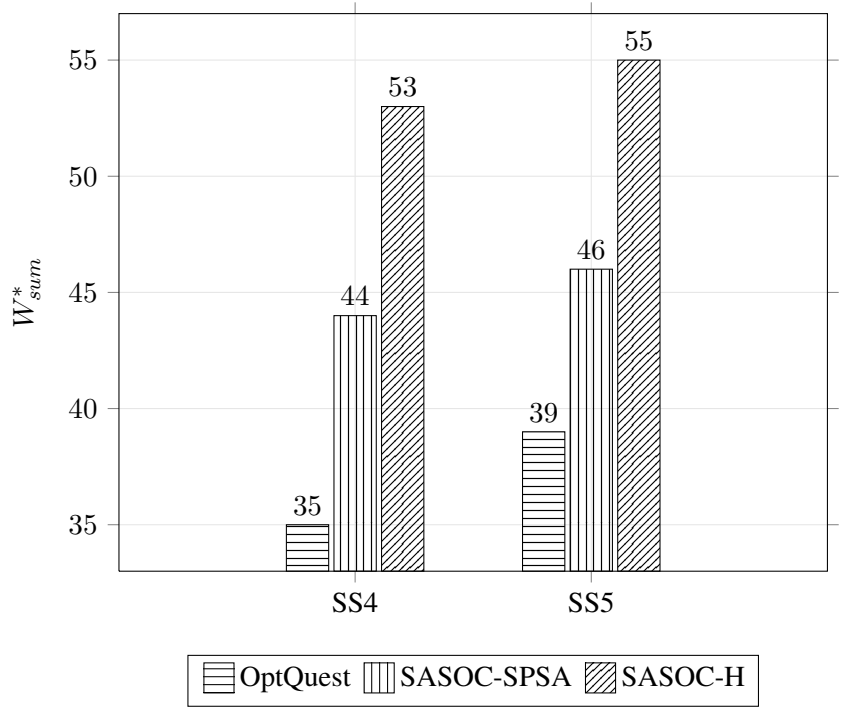
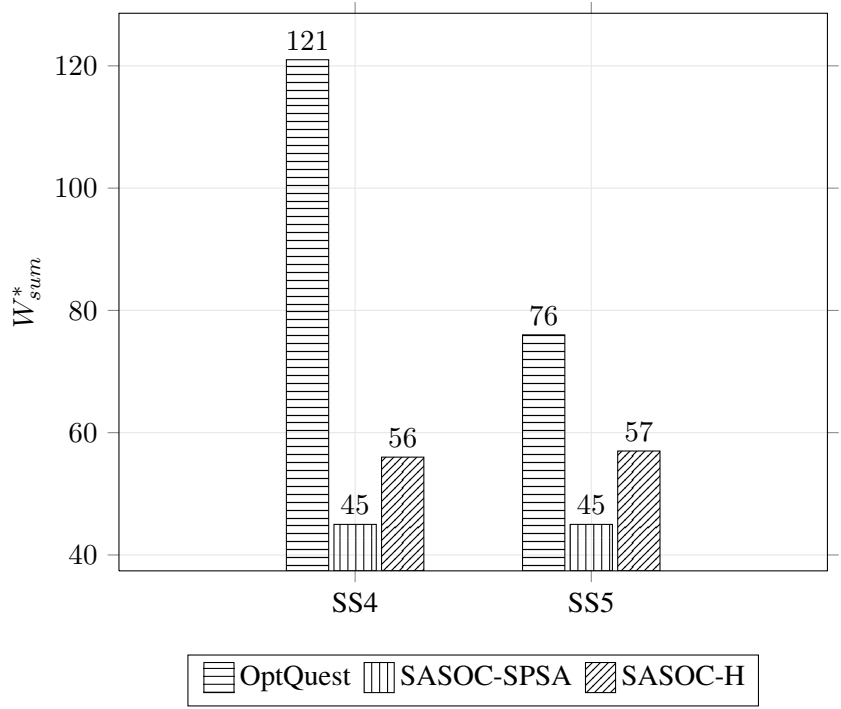

\begin{figure}
\begin{minipage}[c][\textheight]{\textwidth}
    \centering
\tabl{c}{\scalebox{0.8}{\begin{tikzpicture}
\begin{axis}[
ybar={2pt},
legend style={at={(0.5,-0.15)},anchor=north,legend columns=-1},
legend image code/.code={\path[fill=white,white] (-2mm,-2mm) rectangle
(-3mm,2mm); \path[fill=white,white] (-2mm,-2mm) rectangle (2mm,-3mm); \draw
(-2mm,-2mm) rectangle (2mm,2mm);},
ylabel={$W_{sum}^*$},
symbolic x coords={1, SS1, SS4, 2},
xmin={1},
xmax={2},
xtick=data,
ytick align=outside,
bar width=16pt,
nodes near coords,
%nodes near coords align={vertical},
grid,
grid style={gray!20},
width=16cm,
height=11cm,
]
\addplot[pattern=horizontal lines]   coordinates {(SS1,68) (SS4,45) }; %SASOC-SPSA
\addplot[pattern=vertical lines] coordinates {(SS1,63) (SS4,56) }; % SASOC-H
\addplot[pattern=grid] coordinates {(SS1,67) (SS4,80) }; %SASOC-SF-N
\addplot[pattern=north east lines] coordinates {(SS1,63) (SS4,74) }; %SASOC-SF-C

\legend{SASOC-SPSA, SASOC-H, SASOC-SF-N, SASOC-SF-C}
\end{axis}
\end{tikzpicture}}\\[1ex]}
\label{fig:sf-compare}
\caption[Performance comparison of SASOC algorithms with smoothed functional algorithms]{Performance comparison of SASOC with smoothed functional
algorithms. \footnote{Note: the underlying dispatching policy is EDF.}}
\end{minipage}
\end{figure}
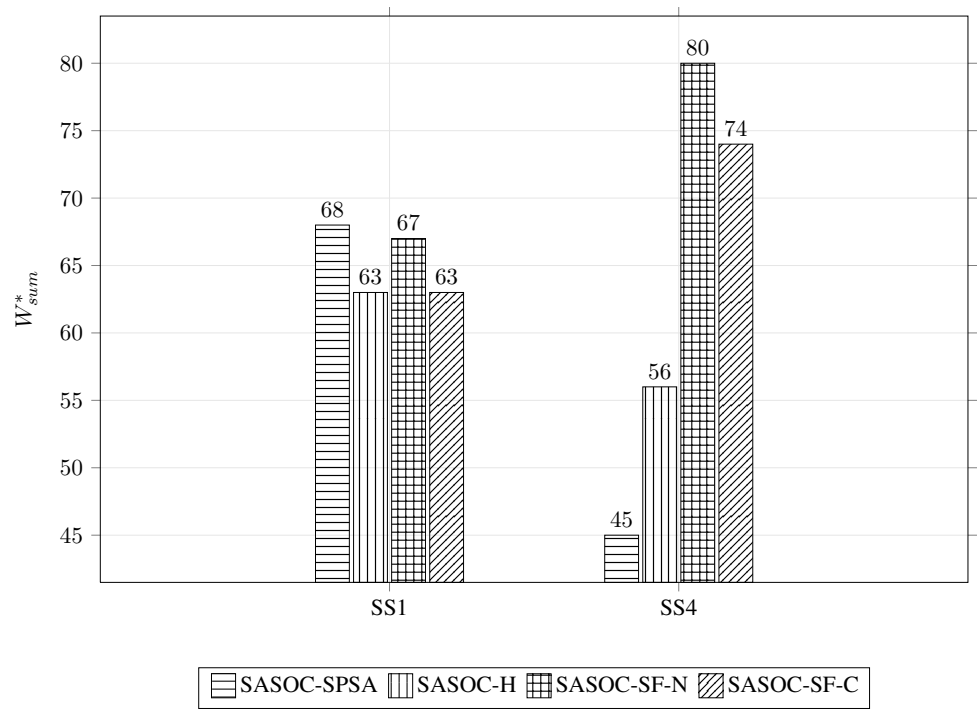

\begin{figure}
    \centering
    \begin{tabular}{cc}
    \subfigure[Using PRIO-PULL on SS4]
    {
%         \hspace{-5em}
\includegraphics[width=3.7in,angle=270]{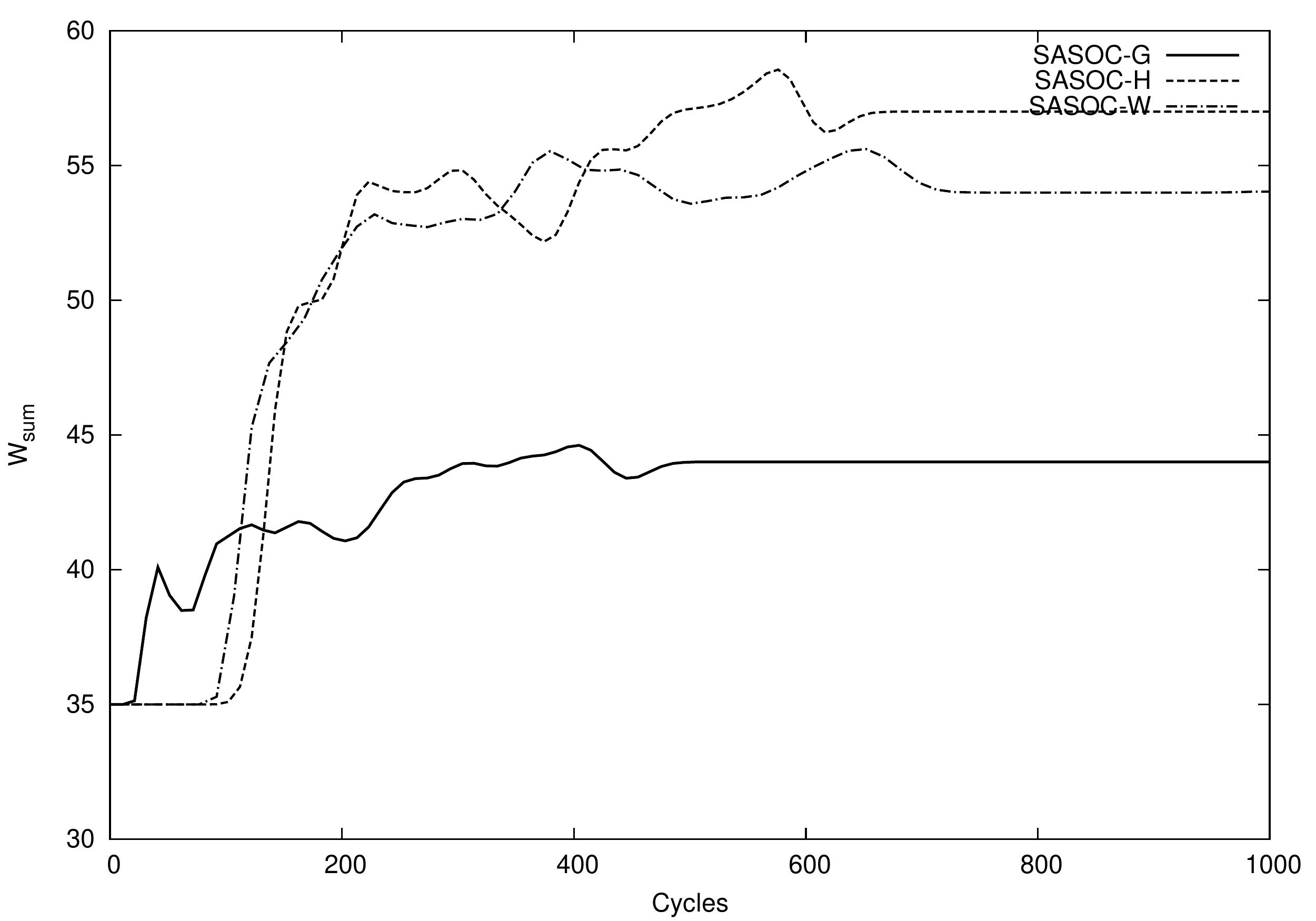}
        \label{fig_bra06_pp}
   }&\\
   \subfigure[Using EDF on SS1]
   {
%        \hspace{-2em}
\includegraphics[width=3.7in,angle=270]{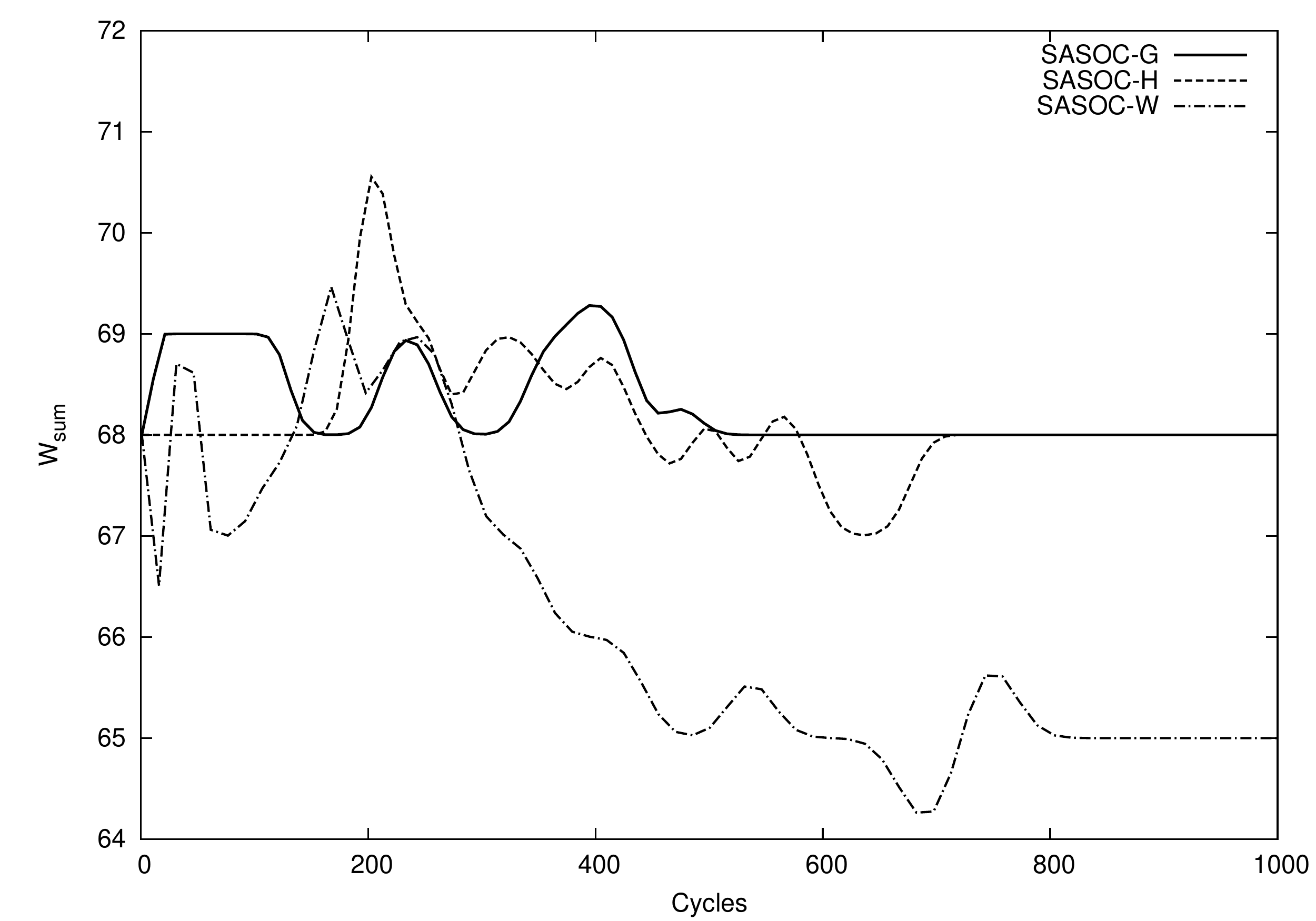}
       \label{fig_arg05_edf}
   }
    \end{tabular}
    \caption{Convergence of $W_{sum}$ as a function of number of cycles for
different SASOC algorithms - Illustration on SS1 and SS4 for two dispatching
policies}
    \label{fig_wsum_conv}
\end{figure}

\begin{figure}
    \centering
\tabl{c}{\scalebox{0.8}{\begin{tikzpicture}
\begin{axis}[
ybar={2pt},
legend style={at={(0.5,-0.15)},anchor=north,legend columns=-1},
legend image code/.code={\path[fill=white,white] (-2mm,-2mm) rectangle
(-3mm,2mm); \path[fill=white,white] (-2mm,-2mm) rectangle (2mm,-3mm); \draw
(-2mm,-2mm) rectangle (2mm,2mm);},
ylabel={$W_{sum}^*$},
symbolic x coords={1, SS1, SS2, SS3, SS4, SS5, 2},
xmin={1},
xmax={2},
xtick=data,
ytick align=outside,
bar width=14pt,
nodes near coords,
%nodes near coords align={vertical},
grid,
grid style={gray!20},
width=16cm,
height=11cm,
]
\addplot[pattern=horizontal lines] coordinates {(SS1,24) (SS2,0) (SS3,95) (SS4,31) (SS5,54)    }; % OptQuest
\addplot[pattern=vertical lines]   coordinates {(SS1,54) (SS2,39) (SS3,100) (SS4,96) (SS5,91)  }; % SASOC-SPSA
\addplot[pattern=grid]             coordinates {(SS1,100) (SS2,70) (SS3,100) (SS4,89) (SS5,68) }; % SASOC-H
\legend{OptQuest, SASOC-SPSA, SASOC-H}
\end{axis}
\end{tikzpicture}}\\[1ex]}
\label{fig:mean-util}
\caption{Performance of OptQuest and SASOC for EDF dispatching policy. The mean utilization values have been rounded to nearest integer.}
\end{figure}
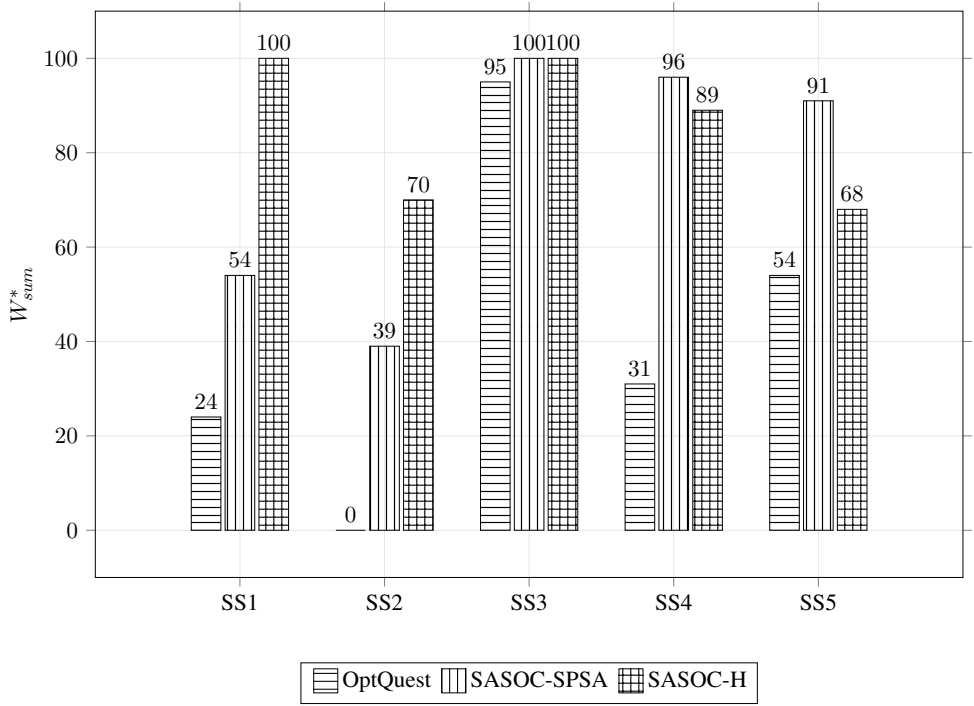

Figures \ref{fig_wsum_priopull_ss123} and \ref{fig_wsum_edf_ss123} compare the
$W^*_{sum}$ achieved for OptQuest and SASOC algorithms using PRIO-PULL
and EDF on three real life SS with a flat SR arrival pattern (see Figure \ref{fig_arrivals_ss123}). 
Here $W^*_{sum}$ denotes the value
obtained upon convergence of $W_{sum}$. On these SS pools, namely
SS1, SS2 and SS3, respectively, we observe that our SASOC algorithms
find a better value of $W^*_{sum}$ as compared to OptQuest.
Note in particular that on SS1, SASOC algorithms perform significantly better than OptQuest with an improvement of nearly $100\%$.
Further, on SS2, OptQuest is seen to be infeasible whereas all the
SASOC algorithms obtain a feasible and good allocation.

It is evident that SASOC algorithms consistently outperform the OptQuest algorithm on these SS pools.
Further, among the SASOC algorithms, we observe that SASOC-W finds
better solutions in general as compared to the other two SASOC
algorithms. Further, we observe that in all our experiments that include both flat as well as bursty arrival pools, the optimal worker parameter obtained by all our SASOC algorithms is feasible, i.e., satisfies both the SLA as well as the queue stability constraints.

Figure \ref{fig_wsum_edf_ss123} presents similar results for the case of the
EDF dispatching policy. The behavior of OptQuest and SASOC algorithms
was found to be similar to that of PRIO-PULL with SASOC showing
performance improvements over OptQuest here as well.
% We present the
% utilization percentages across different skill levels (low, medium and
% high) in Figs. \ref{fig_util_priopull_ss123} and \ref{fig_util_edf_ss123}. We observe mean utilization of workers is a crucial factor for a labor staffing algorithm and it is evident from Figs. \ref{fig_util_priopull_ss123} and \ref{fig_util_edf_ss123} that SASOC algorithms exhibit a higher mean utilization of workers and hence, better overall performance in comparison to the OptQuest algorithm.

\subsection{Bursty-Arrival SS pools}

Figures \ref{fig_wsum_priopull_ss45} and \ref{fig_wsum_edf_ss45} compare the
$W^*_{sum}$ achieved for OptQuest and SASOC algorithms using PRIO-PULL
and EDF on two real life SS with a bursty SR arrival pattern (see Figure \ref{fig_arrivals_ss45}).
% The utilization percentages for these pools are presented in Figures \ref{fig_util_priopull_ss45} and \ref{fig_util_edf_ss45}.
From these performance plots, we observe that OptQuest is seen to be slightly better than SASOC-G and
SASOC-W when the underlying dispatching policy is PRIO-PULL, whereas in the case of EDF dispatching policy,
the SASOC algorithms clearly outperform OptQuest. The execution time advantage of SASOC algorithms over OptQuest hold in the case of these pools as well.

Computational efficiency is a significant factor for any adaptive labor staffing algorithm. For instance, if a candidate labor staffing algorithm takes too long to find the optimal staffing levels, it is not amenable for making staffing changes in a real SS. Both from the number of simulations required as well as the wall clock run time standpoints, SASOC algorithms are better than OptQuest. This is because
OptQuest requires $5000$ iterations with each iteration of $100$
replications, whereas the SASOC algorithms require $1000$
iterations of $20$ replications each in order to find $W^*_{sum}$.
This results in a $25X$ speedup for SASOC algorithms and also manifests in the wall clock runtimes of SASOC algorithms because simulation run-times are
proportional to the number of SS simulations. We observe that the SASOC
algorithms result in at least $10$ to $15$ times improvement as compared to OptQuest from the wall clock runtimes perspective. For
instance, on SS1 the typical run-time of OptQuest was found to be $24$
hours, whereas SASOC algorithms took less than $2.5$ hours each to converge.

In fact, we observed in the case of SS2, OptQuest does not find a feasible solution even after repeated runs  for $5000$ search iterations. Also, because OptQuest
depends heavily on SLA attainments and respective confidence intervals
of previous iterations, it requires higher number of replications than
SASOC.  Further, we observed that SASOC algorithms converge within $500$ iterations in all our experiments. Thus, SASOC algorithms require $25$ times less number of simulations as compared to OptQuest, while searching for the optimal SS
configuration. This runtime advantage ensures that an SS manager can make staffing changes even at the granularity of every week by making use of SASOC algorithms and the same may not be possible with OptQuest due to its longer runtimes. 

\subsection{Comparsion with SF approaches}
Figure \ref{fig:sf-compare} compares the $W^*_{sum}$ achieved with EDF as the
dispatching policy for the SASOC algorithms with the smoothed functional (SF)
based schemes from \citep{prashanth2011ss}. We observe that the SASOC algorithms
perform on par with the Cauchy variant (SASOC-SF-C), while performing better
than the Gaussian variant of the algorithm from \citep{prashanth2011ss}. An
important advantage with our SASOC algorithms in comparison with the SF based
approaches, especially the Cauchy variant, is the low computational overhead. While
our algorithms require Bernoulli random variable for perturbing the worker
parameter, the SF approaches require Gaussian or Cauchy random variables for
the same. Further, the second order method that we propose here (SASOC-H) is
more robust in comparison to the first order SF approaches and through the use
of Woodbury's identity, we also achieve low computational overhead as well.

\subsection{Empirical Convergence of $\theta$}
We observe that the parameter $\theta$ (and hence $W_{sum}$)
converges to the optimum value for each of the SS pools
considered. This is illustrated by the convergence plots in
Figures \ref{fig_bra06_pp} and \ref{fig_arg05_edf}. This is a
significant feature of SASOC as our algorithms are seen to converge analytically
(see Section \ref{sec:convergence}) and the plots confirm the
same. In contrast, the OptQuest algorithm is not proven to converge to
the optimum even after repeated runs, as illustrated in the case of SS2 in Figure
\ref{fig_wsum_priopull_ss123}.

\subsection{Mean utilization results}
We present the
utilization percentages across different skill levels (low, medium and high) in 
Figure \ref{fig:mean-util}. The underlying dispatching policy here is EDF. The results for the case of PRIO-PULL are similar. We observe mean utilization of workers is a crucial factor for a labor staffing algorithm and it is evident from Figure \ref{fig:mean-util} that SASOC algorithms exhibit a higher mean utilization of workers and hence, better overall performance in comparison to the OptQuest algorithm.

From the above performance comparisons over SS pools with flat as well as bursty SR arrival patterns,
it is evident that our SASOC
algorithms, which converge to a local saddle point, show overall better performance in comparison with the scatter search-based algorithm of OptQuest. Among the SASOC algorithms, we observe that the second order algorithms (SASOC-H and SASOC-W) perform
better than the first order algorithm (SASOC-G) in many cases, with
SASOC-W being marginally better than SASOC-H.

%%%%%%%%%%%%%%%%%%%%%%%%%%%%%%%%%%%%%%%%%%%%%%%%%%%%%%%%%%%%
\section{Conclusions}
\label{sec:conclusion}
We motivated the discrete optimization problem of adaptively determining optimal staffing
levels in SS and proposed two novel SASOC algorithms for
solving this problem.  The aim was to find an optimum
worker parameter that minimizes a certain long-run cost objective, while
adhering to a set of constraint functions, which are also long run averages.
All SASOC algorithms are simulation-based optimization methods as the
single-stage cost and constraint functions are observable only via simulation
and no closed form expressions are available.
For solving
the constrained optimization problem, we applied the Lagrange
relaxation procedure and used an SPSA based scheme for performing
gradient descent in the primal and at the same time, ascent in the
dual, for the Lagrange multipliers. All SASOC algorithms also incorporated a smooth (generalized) projection operator that helped imitate a continuous parameter system with suitably defined transition dynamics. Using the theory of
multi-timescale stochastic approximation, we presented the convergence
proof of our algorithms.  Numerical experiments were performed to
evaluate each of the algorithms based on real-life SS data against the
state-of-the-art simulation optimization toolkit OptQuest in the
current context.  SASOC algorithms in general showed overall superior performance
compared to OptQuest, as they (a) exhibited more than an order of
magnitude faster convergence than OptQuest, (b) consistently found
solutions of good quality and in most cases better than those found by
OptQuest, and (c) showed guaranteed convergence even in scenarios
where OptQuest did not find feasibility even after repeated runs for $5000$ iterations.
Given the quick convergence of SASOC algorithms (in minutes), they are
particularly suitable for adaptive labor staffing where a few days of
optimization run like in OptQuest would fail to keep up with the
changes.  By comparing the results of the SASOC algorithms on two
independent dispatching policies, we showed that SASOC's performance
is independent of the operational model of SS.

As future work, one may consider single-stage cost function enhancements that
include worker salaries as well as other relevant monetary costs, apart from staff
utilization and SLA attainment factors. 
An orthogonal direction of future work in this context is to
develop skill updation algorithms, i.e., derive novel work dispatch policies
that improve the skills of the workers beyond their current levels by way of
assigning work of a higher complexity. However, the setting is still constrained
and the SLAs would need to be met while improving the skill levels of the
workers. The skill updation scheme could then be combined  with the SASOC
algorithms presented in this paper to optimize the
staffing levels on a slower timescale.

%%%%%%%%%%%%%%%%%%%%%%%%%%%%%%%%%%%%%%%%%%%%%%%%%%%%%%%%%%%%
\newpage
\appendix
\section{Appendix: Convergence Analysis}
%%%%%%%%%%%%%%%%%%%%%%%%%%%%%%%%%%%%%%%%%%%%%%%%%%%%%%%%%%%
% \section*{Convergence analysis}
% \label{sec:convergence}
Here we provide a sketch of the convergence of the SASOC-G and SASOC-H algorithms. The first step in the convergence analysis is common to all the SASOC algorithms and involves the extension of the transition dynamics $p_\theta(i,j)$ of the constrained parameterized hidden Markov process to the convex hull $\bar\D$.

\subsection*{Extension of the transition dynamics $p_{i,j}(\theta)$}
Recall that the discrete parameter $\theta$ of the Markov process $\{X_n(\theta)\}$ takes values in the set $\D$ defined earlier. Using the members of $\D$, one can extend the transition dynamics $p_\theta(i,j)$ of the underlying Markov process to any $\theta$ in the convex hull $\bar\D$ as follows:
% one can write any $\theta \in \bar\D$ as
% \begin{equation}
%  \theta = \sum\limits_{k=1}^{N} \beta_k(\theta) D^k,
% \label{eq:thetaD}  
% \end{equation}
% where the weights $\beta_k(\theta)$ satisfy $0 \le \beta_k(\theta) \le 1, k=1,ldots,N$ and $\sum\limits_{k=1}^{N} \beta_k(\theta) = 1$.
% We now define the transition probability $p_\theta(i,j)$ for any $\theta \in \bar\D$ as
\begin{equation}
     p_\theta(i,j) = \sum\limits_{k=1}^{p} \beta_k(\theta) p_{D^k}(i,j), \quad \forall \theta \in \bar\D, i,j \in S,
\label{eq:pthetabar}
\end{equation}
where the weights $\beta_k(\theta)$ satisfy $0 \le \beta_k(\theta) \le 1, k=1,\ldots,p$ and $\sum\limits_{k=1}^{p} \beta_k(\theta) = 1$. $p_\theta(i,j), i,j \in S, \theta \in \bar\D$ can be seen to satisfy the properties of transition probabilities. It is worth noting here that the weights $\beta_k(\theta)$ must be continuously differentiable in order to ensure that the extended transition probabilities are continuously differentiable as well and our SASOC algorithms converge. Moreover, in the SASOC algorithms, we do not require an explicit computation of these weights while trying to solve the constrained optimization problem equation (1) of the main paper. Consider the case when $\theta = (\theta_1)^T$ and suppose $\theta_1$ lies between $D^j$ and $D^{j+1}$ (both members of $\D$). By construction, $\beta_k(\theta)$ will correspond to the probability with which projection is done on $[D^j,D^{j+1}]$ and is obtained using the $\Gamma$-projection operator as follows:  Let us consider an interval of length $2 \
zeta$ around the midpoint of $[D^j, D^{j+1}]$ and denote it as $[\tilde D_1, \tilde D_2]$, where $\tilde D_1 = \frac{\D^j+\D^{j+1}}{2} - \zeta$ and $\tilde D_2 = \frac{\D^j+\D^{j+1}}{2} + \zeta$.
Then, the weights $\beta_k(\theta)$ are set in the following manner:
$\beta_k(\theta) = 0, \forall k \notin \{j,j+1\}$ and $\beta_j(\theta), \beta_{j+1}(\theta)$ is given by: 

\begin{equation}
 (\beta_j(\theta_1),\beta_{j+1}(\theta_1)) =
  \begin{cases}
   (1,0) &  \text{if } \theta \in\left[D^j,\tilde D_1 \right] \\
   (f(\frac{\tilde D_2 - \theta_1}{2\zeta}), 1- f(\frac{\tilde D_2 - \theta_1}{2\zeta})) & \text{if } \theta_1 \in \left[\tilde D_1, \tilde D_2\right]  \\
   (0,1) & \text{if } \theta \in \left[\tilde D_2,D^{j+1}\right]
  \end{cases}
\end{equation}
In the above, $f$ is obtained from the definition of $\Gamma$-projection and hence, is a continuously differentiable function defined on $[0,1]$ such that $f(0)=0$ and $f(1)=1$. The above can be similarly extended when the parameter $\theta$ has $N$ components. It can thus be seen that $\beta_k(\theta), k=1,\ldots,p$ are continuously differentiable functions of $\theta$. Thus, from  \eqref{eq:pthetabar} and the fact that $\beta_k(\theta)$ are continuously differentiable, it can be seen that the extended transition dynamics $p_\theta(i,j), \forall \theta \in \bar\D, i,j \in S$ are continuously differentiable.

We now claim the following:
% \begin{lemma}
%     For any $\theta \in \bar\D$, $i,j \in S$ and for any $n\ge 1$ 
%      \[p^n_\theta(i,j) \ge \sum\limits_{k=1}^{N} \beta^n_k(\theta) p^n_{D^k}(i,j).\]
% \end{lemma}
% \begin{proof}
%     Follows in a similar manner as Lemma 1 of \cite{shalabh2011stochastic}.
% \end{proof}

\begin{lemma}
\label{lemma:markov}
    For any $\theta \in \bar\D$, $\{ X_n(\theta), n\ge 0\}$ is ergodic Markov.
\end{lemma}
\begin{proof}
    Follows in a similar manner as Lemma 2 of \cite{shalabh2011stochastic}.
\end{proof}

Now, define analogues of the long-run average cost and constraint functions for any $\theta \in \bar\D$ as follows:
\begin{equation}
\label{eqn:c-mpbar}
\begin{array}{l}
\bar J(\theta) \stackrel{\triangle}{=} \lim\limits_{n \rightarrow \infty}\frac{1}{n} \sum\limits_{m=0}^{n-1} c(X_m(\theta)), \theta \in \bar\D\\
\bar G_{i,j}(\theta) \stackrel{\triangle}{=} \lim\limits_{n \rightarrow \infty}\frac{1}{n} \sum\limits_{m=0}^{n-1} g_{i,j}(X_m(\theta)) \le 0, \\\qquad\qquad\qquad\quad\forall i=1,\ldots,|C|, j=1,\ldots,|P|, \theta \in \bar\D \\
\bar H(\theta) \stackrel{\triangle}{=} \lim\limits_{n \rightarrow \infty}\frac{1}{n} \sum\limits_{m=0}^{n-1} h(X_m(\theta)) \le 0, \theta \in \bar\D.
%g_{i,j}(X_m) =  \le 0 \quad\forall i=1,\ldots,|C|, j=1,\ldots,|P|
\end{array}
\end{equation}
The difference between the above and the corresponding entitites defined in equation (1) of the main paper is that $\theta$ can take values in $\bar\D$ in the above. In lieu of Lemma 2, the above limits are well-defined for all $\theta \in \bar\D$.

\begin{lemma}
$\bar J(\theta), \bar G_{i,j}(\theta), i=1,\ldots,|C|, j=1,\ldots,|P|,$ and $\bar H(\theta)$ are continuously differentiable in $\theta \in \bar \D$.
\end{lemma}
\begin{proof}
    Follows in a similar manner as Lemma 3 of \cite{shalabh2011stochastic}.
\end{proof}

We now prove the SASOC algorithms described previously are equivalent to their analogous continuous parameter $\bar\theta$ counterparts under the extended Markov process dynamics.
 
\begin{lemma}
\label{lemma:sasocequivalence}
Under the extended dynamics $p_\theta(i,j),i,j \in S$ of the Markov process $\{X_n(\theta)\}$ defined over all $\theta \in \bar \D$, we have
\begin{enumerate}[(i)]
    \item SASOC-G algorithm is analogous to its continuous counterpart where $\Gamma(\theta)$ and $\Gamma(\theta + \delta \triangle)$ are replaced by $\bar\Gamma(\theta)$ and $\bar\Gamma(\theta + \delta \triangle)$ respectively.
    \item SASOC-H algorithm is analogous to its continuous counterparts where $\Gamma(\theta)$ and $\Gamma(\theta + \delta_1 \triangle + \delta_2 \hat\triangle)$ are replaced by $\bar\Gamma(\theta)$ and $\bar\Gamma(\theta +\delta_1 \triangle + \delta_2 \hat\triangle)$ respectively.
\end{enumerate}

\end{lemma}
\begin{proof}
\textbf{(i)}:
Consider the SASOC-G algorithm which updates according to  equation (11) of the main paper. Let $\theta(m)$
be a given parameter update that lies in $\bar{D}^o$ (where $\bar{D}^o$ denotes
the interior of the set $\bar{D}$). Let $\delta>0$ be sufficiently small so that
$\bar{\theta}^1(m)=(\bar{\Gamma}_j(\theta_j(m)+\delta\Delta_j(m)),j=1,\ldots,N)^T$
$=(\theta_j(m)+\delta\Delta_j(m)),j=1,\ldots,N)^T$.

Consider now the $\Gamma$-projected parameters
$\theta^1(m) = (\Gamma_j(\theta_j(m)+\delta\Delta_j(m)),j=1,\ldots,N)^T$ and
$\theta^2(m) = (\Gamma_j(\theta_j(m)),j=1,\ldots,N)^T$, respectively.
By the construction of the generalized projection operator,
these parameters are equal to $\theta^k\in C$ with probabilities
$\beta_k((\theta_j(m)+\delta\Delta_j(m),j=1,\ldots,N)^T)$
and
$\beta_k(\theta_j(m),j=1,\ldots,N)^T)$, respectively.
When the operative parameter is $\theta^k$, the transition probabilities are
$p_{\theta^k}(i,l)$, $i,l\in S$. Thus with probabilities
$\beta_k((\theta_j(m)+\delta\Delta_j(m),j=1,\ldots,N)^T)$
and
$\beta_k((\theta_j(m),j=1,\ldots,N)^T)$, respectively,
the transition probabilities in the two simulations equal $p_{\theta^k}(i,l)$,
$i,l\in S$.

Next, consider the alternative (extended) system with parameters $\bar{\theta}^1(m) = (\bar\Gamma_j(\theta_j(m)+\delta\Delta_j(m))$
and $\bar{\theta}^2(m) = \bar\Gamma_j(\theta_j(m))$, respectively. The transition probabilities are now given by
\[p_{\bar{\theta}^i(m)}(j,l) = \sum_{k=1}^{p} \beta_k(\bar{\theta}^i(m))p_{\theta^k}(j,l),\]
$i=1,2$, $j,l\in S$. Thus with probability $\beta_k(\bar{\theta}^i(m))$, a transition
probability of $p_{\theta^k}(j,l)$ is obtained in the $i$th system.
Thus the two systems (original and the one with extended dynamics) are analogous.

Now consider the case
when $\theta(m)\in \partial\bar{\D}$, i.e., is a point on the boundary of $\bar{\D}$).
Then, one or more components of $\theta(m)$ are extreme points. For simplicity,
assume that only one component (say the $i$th component) is an extreme point as the
same argument carries over if there are more parameter components that are extreme
points. By
the $i$th component of $\theta(m)$ being an extreme
point, we mean that $\theta_i(m)$ is either $0$ or $W_{\max}$.
The other components $j=1,\ldots,N,j\not=i$ are not extreme.
Thus, $\theta_i(m)+\delta\Delta_i(m)$ can lie outside of the interval $[0,W_{\max}]$. For instance, suppose
that $\theta_i(m)=W_{\max}$ and that $\theta_i(m)+\delta\Delta_i(m)>W_{\max}$
(which will happen if $\Delta_i(m)=+1$). In such a case,
$\theta^1_i(m) = \Gamma_i(\theta_i(m)+\delta\Delta_i(m)) = W_{\max}$ with probability one.
Then, as before, $\theta^1(m)$ can be written as the convex combination
${\displaystyle \theta^1(m) = \sum_{k=1}^{p}\beta_k(\theta^1(m))
\theta^k}$ and
the rest follows as before.

\textbf{(ii)}: Follows in a similar manner as part (i) above. 
\end{proof}

As a consequence of Lemma \ref{lemma:sasocequivalence}, we can analyze the SASOC algorithms with the continuous parameter $\bar\theta$ used in place of $\theta$ and under the extended transition dynamics \eqref{eq:pthetabar}. By an abuse of notation, we shall henceforth use $\theta$ to refer to the latter. 

\subsection*{SASOC-G}
 The convergence analysis of SASOC-G can be split into four stages:

\begin{inparaenum}[\bfseries (I)]
\hspace{-1em}\item The fastest time-scale in SASOC-G is $\{d(n)\}$ which is used to update the Lagrangian estimates $\bar{L}$ and $\bar{L}'$ corresponding to simulations with $\theta$ and $\theta + \delta \Delta$ respectively. Firstly, we show that these estimates indeed converge to the Lagrangian values $L(\theta, \lambda)$ and $L(\theta + \delta \Delta, \lambda)$ defined in equation (9) of the main paper. Note that the $\theta$ and $\lambda$ which are updated on slower time-scales, can be assumed to be time invariant quantities for the purpose of analysis of these Lagrangian estimates.\\
\item Next, we show that  the parameter updates $\theta(n)$ using SASOC-G converge to a limit point of the ODE
\begin{equation}
\label{eqn:sasoc-g:theta-ode}
\dot{\theta}(t) = \check{\Gamma}\left ( -\nabla_\theta L(\theta(t), \lambda) \right ),
\end{equation}
where $\check{\Gamma}$ is defined as follows: For any bounded continuous function $\epsilon(\cdot)$,
\begin{equation}
\label{eqn:Pi-bar-operator}
\check{\Gamma}(\epsilon(\theta(t))) = \lim\limits_{\eta \downarrow 0} \dfrac{\Pi(\theta(t) + \eta \epsilon(\theta(t))) - \theta(t)}{\eta}.
\end{equation}
The projection operator $\check{\Gamma}(\cdot)$ ensures that the evolution of $\theta$ stays within the bounded set $M$. Again for the analysis of the $\theta$-update, the value of $\lambda$ which is updated on the slowest time-scale is assumed constant.\\
%We show below that the evolution of $\theta$ in SASOC-G descends in the Lagrangian value and converges to a limiting set that depends on $\lambda$. For this purpose, we first show that the resulting martingale from the $\theta$ update recursion in (\ref{eqn:spsa-update-rule}) is convergent and then use $V^{\lambda}(\cdot) = L (\theta,\lambda)$ as an associated Lyapunov function for the ODE (\ref{ode}).
\item We show that
$\lambda_{i, j}$s and $\lambda_f$ converge respectively
to the limit points of the ODEs \[\begin{array}{l}
\dot{\lambda}_{i,j}(t) = \check\Pi \left ( G_{i, j}(\theta^*) \right ), \forall
i
=
1, 2, \dots, |C|, j = 1, 2, \dots, |P|,\\[1ex]
\dot{\lambda}_f(t) = \check\Pi \left ( H(\theta^*) \right ),
\end{array}\]
where $\theta^*$ is the converged parameter value of SASOC-G/H corresponding to
Lagrange parameter $\lambda(t) \stackrel{\triangle}{=} (\lambda_{i,j}(t),
\lambda_f(t), i=1,\ldots,|C|, j=1,\ldots,|P|)^T$, and for any bounded continuous
functions $\bar{\epsilon}(\cdot)$, \[\check\Pi(\bar{\epsilon}(\lambda(t))) =
\lim\limits_{\eta \downarrow 0} \dfrac{(\lambda(t) + \eta
\bar{\epsilon}(\lambda(t)))^+ - \lambda(t)}{\eta}.\] 
Here again, the projection
operator $\check\Pi$ ensures that the evolution of each component
of $\lambda$ stays non-negative. From the definition of the Lagrangian given in equation (9) of the  main paper, the gradient of the Lagrangian w.r.t. $\lambda_{i, j}$ can be seen to  be $G_{i,j}(\theta^*)$ and that w.r.t. $\lambda_f$ to be $H(\theta^*)$. Thus, the above ODEs suggest that in SASOC-G $\lambda_{i, j}$s and $\lambda_f$ are ascending in the Lagrangian value and converge to a local maximum point.\\
\item Finally, we show that the algorithm indeed converges to a (local) saddle point of the Lagrangian with local maximum in $\lambda_{i, j}$s and $\lambda_f$, and local minimum in $\theta$.
\end{inparaenum}

\begin{lemma}
\label{lemma:Lagrangian}
$\|\bar{L}(n) - L(\theta(n), \lambda(n)) \| \rightarrow 0$ w.p. 1, as $n \rightarrow \infty$.

\begin{proof}
{\rm $\theta$ and $\lambda$ values are being updated on slower time-scales, thus assumed to be constant in this proof. Let \[l(X_m) \stackrel{\triangle}{=} c(X_{nK+m}) + \sum\limits_{i=1}^{|C|}\sum\limits_{j=1}^{|P|} \lambda_{i,j}(nK) g_{i,j}(X_{nK+m}) + \lambda_f h(X_{nK+m}).\]
The $\bar{L}$ update can be re-written as
\[\bar{L}(m + 1) = \bar{L}(m) + d(m)\left ( L(\theta(m),\lambda(m)) + \xi_1(m) - \bar{L}(m) + M_{m + 1} \right ), \]
where $\xi_1(m) = (E[ l(X_m) | \mathcal{F}_{m - 1} ] - - L(\theta(m),\lambda(m)), m \ge 0$ and \\$\mathcal{F}_m = \sigma(X_n, \lambda(n), \theta(n), n \le m), m \ge 0$ are the associated $\sigma$-fields. Also, $M_{m + 1} = l(X_m) - E[ l(X_m) | \mathcal{F}_{m - 1} ], m \ge 0$ is a martingale difference sequence. Let $N_m = \sum_{n = 0}^{m} d(n) M_{n + 1}$. It can be easily verified that $(N_m, \mathcal{F}_m), m \ge 0$ is a square-integrable martingale obtained from the corresponding martingale difference $\{M_m\}$. Further, from the square summability of $d(n), n \ge 0$, and the facts that $S$ is compact and $l$ is Lipschitz continuous, it can be verified from the martingale convergence theorem that $\{N_m, m \ge 0\}$, converges almost surely.

Now from Lemma \ref{lemma:markov} $\{(X_m)\}$ is ergodic Markov for any given $\theta(m)$. Hence,
$|E[ l(X_m) | \mathcal{F}_{m - 1} ] - - L(\theta(m),\lambda(m)| \rightarrow 0$ almost surely on the `natural timescale', as $m \rightarrow \infty$.  The `natural timescale' is clearly faster than the algorithm's timescale and hence $\xi_1(m)$ can be ignored in the analysis of $\bar L$-recursion, see \cite[Chapter 6.2]{borkar2008stochastic} for detailed treatment of natural timescale algorithms.
The rest of the proof follows from the Hirsch lemma \cite[Theorem 1, pp. 339]{hirsch1989convergent}.
}
\end{proof}
\end{lemma}

On similar lines, $\|\bar{L}'(n) - L(\theta(n) + \delta \Delta(n), \lambda(n)) \| \rightarrow 0$ w.p. 1, as $n \rightarrow \infty$.  Thus, $\theta$ updates which are on the slower time scale $\{b(n)\}$, can be re-written as
\begin{equation}
W_{i}(n + 1) = W_{i}(n) - b(n) \left ( \frac{L(\theta(n) + \delta \Delta_{i}(n), \lambda) - L(\theta(n), \lambda)}{\delta \Delta_{i}(n)}  \right) + b(n) \chi_{n+1} ,
\label{weq}
\end{equation}
$\forall i = 1, 2, \dots, |A| \times |B|$, where $\chi_n = o(1)$ in view of Lemma \ref{lemma:Lagrangian}. Note here that $\lambda(n) \equiv \lambda, \quad \forall n$.
Now for the ODE (\ref{eqn:sasoc-g:theta-ode}), $V^{\lambda}(\cdot) = L(\cdot,\lambda)$ serves as an associated Lyapunov function and the stable fixed points of this ODE lie within the set $K^{\lambda} = \{ \theta \in S: \check{\Gamma}\left ( -\nabla L(\theta(t), \lambda) \right ) = 0 \}$.

\begin{theorem}
Under (A1)-(A3), in the limit as $\delta \rightarrow 0$,
$\theta(R) \rightarrow \theta^* \in K^{\lambda}$ almost surely as $R \rightarrow \infty$.
\begin{proof}
{\rm
%Follows in a similar manner as Theorem 2.4 of \cite{bhatnagar2003two}.
From assumption (A2), $L(\theta, \lambda)$ is assumed to be continuous. Hence over the compact set $M$, $L(\theta, \lambda)$ is uniformly bounded. Thus, from Lasalle's invariance theorem \cite{lasalle} \cite[Theorem 2.3, pp. 76]{kushner-yin}, $\theta(R) \rightarrow \theta^* \in K^{\lambda}$ a.s. as $R \rightarrow \infty$.
}
\end{proof}
\end{theorem}

Thus (\ref{weq}) can be seen to be an Euler discretization of (\ref{eqn:sasoc-g:theta-ode}) and converges a.s. to $K^{\lambda}$ in the limit as $\delta \rightarrow 0$.

For $\{\lambda(n)\}$ updates on the slowest time-scale $\{a(n)\}$, we can assume that $\theta$ has converged to $\theta^* \in K^\lambda$. Let \[F^{\theta^*} = \left \{ \lambda \ge 0 : \check\Pi \left ( G_{i, j}(\theta^*) \right ) = 0, \forall i = 1, 2, \dots, |C|, j = 1, 2, \dots, |P|; \check\Pi \left ( H(\theta^*) \right ) = 0 \right \}.\]

%  To show convergence of the parameter update sequence $\{\theta_n\}, n \ge 0$, let $K = \{ \theta \in M: \check{\Gamma}\left ( -\nabla L(\theta(t), \lambda) \right ) = 0 \}$, and for some $\eta > 0$, $K^\eta = \{ \theta \in M: \| \theta - \theta_0 \| < \eta\text{ for some }\theta_0 \in K\}$.
%
% \begin{theorem}
% Under assumptions \textbf{(A1)-(A3)}, given $\eta > 0$, there exists $\delta_0 > 0$ such that for all $\delta \in (0, \delta_0]$, the sequence $\{\theta_n\}$ converges to $K^\eta$ a.s.
%
% \begin{proof}
% {\rm
% \hfill $\blacksquare$}
% \end{proof}
% \end{theorem}

\begin{theorem}
\label{theorem:sasoc-g-lambda}
$\lambda(R) \rightarrow \lambda^* \in F$ w.p. 1 as $R \rightarrow \infty$.
\begin{proof}
{\rm
The $\lambda$ update in equation (11) of the main paper can be re-written as
\[\lambda_{i, j}(n + 1) = \lambda_{i, j}(n) + a(n) \left [ G_{i, j}(\theta^*) + N_{n + 1} + M_{n + 1} \right ],\]
where $N_{n + 1} = E[ g_{i, j}(X_n) | \mathcal{F}_{n - 1} ] - G_{i, j}(\theta^*)$, $M_{n + 1} = g_{i, j}(X_n) - E[ g_{i, j}(X_n) | \mathcal{F}_{n - 1} ]$. It is easy to see that from Lemma \ref{lemma:markov} that $N_{n} \rightarrow 0$ as $n \rightarrow \infty$ along the natural timescale (see Lemma \ref{lemma:Lagrangian}).
Further, $\{M_n\}$  is a martingale difference sequence with $\sum_{i = 0}^{n} a(i) M_{i + 1}, n \ge 0$, being the associated martingale that can be seen to be a.s. convergent (See Prop. 4.4 of \cite{shalabh2011constrained}).
 Thus from \cite[Extension 3 of Section 2.2]{borkar2008stochastic}, the result follows for $\lambda_{i, j}$s. Similarly, one can show convergence for $\lambda_f$.
}
\end{proof}
\end{theorem}

Now, we need to show that the convergence of the algorithm is indeed to a saddle point, i.e., $\theta^* \in K^{\lambda^*}$ and $\lambda^* \in F^{\theta^*}$. This can be shown by invoking the envelope theorem of mathematical economics \cite[pp. 964-966]{mas1995microeconomic}; see remark (2) in \cite[pp 15]{shalabh2011constrained}.

\subsection*{SASOC-H}
Convergence analysis of SASOC-H follows along similar lines as that of the SASOC-G algorithm as given below. Note that we first analyse the case when the Hessian is inverted directly in SASOC-H and then give the necessary modifications for the proof to work when Woodbury's identity is employed.
\begin{enumerate}
\item As in Lemma \ref{lemma:Lagrangian}, one can see that $\bar{L}$ and $\bar{L}'$ iterations converge almost surely as follows:
\begin{align*}
\hspace{-2em}\|\bar{L}(n) - L(\theta(n), \lambda(n)) \|, \|\bar{L}'(n) - L(\theta(n) + \delta_1 \Delta(n) + \delta_2 \widehat\Delta(n), \lambda(n)) \| \rightarrow 0 \textrm{ as } n \rightarrow \infty.
\end{align*}
\item Next, we show that the parameter updates $\theta(n)$ of SASOC-H converge to a limit point of the ODE
\begin{equation}
\label{eqn:sasoc-h:theta-ode}
\dot{\theta}(t) = \check{\Gamma}\left ( - \Upsilon(\nabla^2_\theta L(\theta(t), \lambda))^{-1} \nabla_\theta L(\theta(t), \lambda) \right ),
\end{equation}
where $\check{\Gamma}$ is as defined in equation (\ref{eqn:Pi-bar-operator}).
\item The rest of the analysis of slower time-scale updates of $\lambda_{i, j}$s and $\lambda_f$, and saddle point behaviour follows from that of SASOC-G.
\end{enumerate}

\begin{lemma}
\label{lemma:sasoc-h:gradient}
\[\left \| \dfrac{L(\theta(n) + \delta_1 \Delta(n) + \delta_2 \widehat\Delta(n), \lambda(n)) - L(\theta(n),\lambda(n))}{\delta_2 \widehat\Delta_i(n)} - \nabla_{\theta_{i}} L(\theta(n), \lambda(n)) \right \| \rightarrow 0\textrm{ w.p. 1},
\]
with $\delta_1, \delta_2 \rightarrow 0$ as $n \rightarrow \infty \quad \forall i \in \{1, 2, \dots, |A| \times |B| \}$.
\begin{proof}
{\rm
Follows from \cite[Proposition 4.10]{shalabh2011constrained}.
}
\end{proof}
\end{lemma}

\begin{lemma}
\label{lemma:hessian-1}
\[\left \| \dfrac{L(\theta(n) + \delta_1 \Delta(n) + \delta_2 \widehat\Delta(n), \lambda(n)) - L(\theta(n),\lambda(n))}{\delta_1 \Delta_i(n) \delta_2 \widehat\Delta_j(n)} - \nabla^2_{\theta_{i, j}} L(\theta(n), \lambda(n)) \right \| \rightarrow 0\textrm{ w.p. 1},
\]
with $\delta_1, \delta_2 \rightarrow 0$ as $n \rightarrow \infty, \quad \forall i, j \in \{1, 2, \dots, |A| \times |B| \}$.
\begin{proof}
{
\rm
Follows from \cite[Proposition 4.9]{shalabh2011constrained}.
}
\end{proof}
\end{lemma}

\begin{lemma}
\label{lemma:hessian}
\[\left \| H_{i, j}(n) - \nabla^2_{\theta_{i, j}} L(\theta(n), \lambda(n)) \right \| \rightarrow 0\textrm{ w.p. 1},
\]
with $\delta_1, \delta_2 \rightarrow 0$ as $n \rightarrow \infty, \quad \forall i, j \in \{1, 2, \dots, |A| \times |B| \}$.

\begin{proof}
{\rm
Follows from Lemma \ref{lemma:hessian-1} applied to the Hessian update of SASOC-H.
}
\end{proof}
\end{lemma}

\begin{lemma}
\label{lemma:inverse-hessian}
\[\left \| M(n) - \Upsilon(\nabla^2_{\theta} L(\theta(n), \lambda(n)))^{-1} \right \| \rightarrow 0\textrm{ w.p. 1},
\]
with $\delta_1, \delta_2 \rightarrow 0$ as $n \rightarrow \infty, \quad \forall i, j \in \{1, 2, \dots, |A| \times |B| \}$.

\begin{proof}
{\rm
Follows from Lemma \ref{lemma:hessian} and \cite[Lemma A.9]{bhatnagar2007adaptive}.
}
\end{proof}
\end{lemma}

Let \[\bar{K}^\lambda = \left \{ \theta \in S: \dfrac{d L (\theta(t), \lambda)}{dt} = - \nabla_{\theta} L (\theta(t), \lambda)^T \Upsilon(\nabla^2_\theta L(\theta(t), \lambda))^{-1} \nabla_\theta L(\theta(t), \lambda) = 0 \right \}.\]
\begin{theorem}
Under assumptions (A1)-(A4), in the limit as $\delta_1, \delta_2 \rightarrow 0$,
$\theta(R) \rightarrow \theta^* \in \bar{K}^{\lambda}$ almost surely as $R \rightarrow \infty$.
\begin{proof}
{\rm
Following Lemmas \ref{lemma:Lagrangian}, \ref{lemma:sasoc-h:gradient} and \ref{lemma:inverse-hessian}, with $\delta_1, \delta_2 \rightarrow 0$, the update of parameter $\theta$ can we re-written in vector form as
\[ \theta_{n + 1} = \Pi \left ( \theta_n - b(n) \Upsilon(\nabla^2_\theta L(\theta(t), \lambda))^{-1} \nabla_\theta L(\theta(t), \lambda) + b(n) \chi_n \right )\]
with $\chi_n = o(1)$. Thus, the update of parameter $\theta$ can be viewed as a noisy Euler discretization of the ODE (\ref{eqn:sasoc-h:theta-ode}) using a standard approximation argument as in \cite[pp. 191-196]{kushner-clark}. Note that $V^\lambda(\cdot) = L(\cdot, \lambda)$ itself serves as the associated Lyapunov function \cite[pp. 75]{kushner-yin} for the ODE (\ref{eqn:sasoc-h:theta-ode}) with stable limit points of the ODE lying within the set $\bar{K}^\lambda$. From assumption (A2), $L(\theta, \lambda)$ is assumed to be continuous. Hence over the compact set $M$, $L(\theta, \lambda)$ is uniformly bounded. Thus, from Lasalle's invariance theorem \cite{lasalle}, $\theta(n) \rightarrow \theta^* \in K^{\lambda}$ a.s. as $n \rightarrow \infty$.
}
\end{proof}
\end{theorem}

% Rest of the analysis follows along similar lines as that of SASOC-G.

Convergence analysis of SASOC-H when the Hessian in inverted using an iterative procedure based on Woodbury's identity, follows from the above analysis for the SASOC-H algorithm (with direct inversion of the Hessian) given the following lemma instead of Lemma \ref{lemma:inverse-hessian}.

\begin{lemma}
\label{lemma:sasoc-w-inverse-hessian}
\[\left \| M(n) - \Upsilon(\nabla^2_{\theta} L(\theta(n), \lambda(n)))^{-1} \right \| \rightarrow 0\textrm{ w.p. 1},
\]
with $\delta_1, \delta_2 \rightarrow 0$ as $n \rightarrow \infty, \quad \forall i, j \in \{1, 2, \dots, |A| \times |B| \}$.
\begin{proof}
{\rm
From Woodbury's identity, since $M(n), n \ge 1$ sequence of SASOC-W is identical to the $\Upsilon(H(n))^{-1}, n \ge 1$ sequence of SASOC-H, the result follows from Lemma \ref{lemma:inverse-hessian}.
}
\end{proof}
\end{lemma}

\bibliographystyle{plainnat}

\bibliography{sasoc-simulation}

\end{document}